\documentclass[a4paper,USenglish]{lipics-v2021}
\nolinenumbers
\usepackage[utf8]{inputenc}
\usepackage{xcolor}
\usepackage{xspace}
\usepackage{comment}
\usepackage[normalem]{ulem}
\usepackage{bm}

\hideLIPIcs

\makeatletter
\renewcommand{\paragraph}{%
  \@startsection{paragraph}{4}%
  {\z@}{\medskipamount}{-1em}%
  {\noindent\bfseries}%
}
\makeatother

\newcommand{\look}{\texttt{Look}\xspace}
\newcommand{\compute}{\texttt{Compute}\xspace}
\newcommand{\move}{\texttt{Move}\xspace}
\newcommand{\LCM}{{\tt LCM}\xspace}

\newcommand{\fsync}{\textsc{FSync}\xspace}
\newcommand{\ssync}{\textsc{SSync}\xspace}

\newcommand{\async}{\textsc{Async}\xspace}
\newcommand{\nesta}{\textsc{NestA}\xspace}

\newtheorem{obs}{Observation}

\newcommand{\bW}{\mathbb{W}} 
\newcommand{\bX}{\mathbb{X}} 
\newcommand{\bY}{\mathbb{Y}} 
\newcommand{\bZ}{\mathbb{Z}} 

\newcommand{\eps}{\varepsilon}

\newcommand{\tab}{\hspace{10pt}}

\title{Separating Bounded and Unbounded Asynchrony for  Autonomous Robots: Point Convergence with Limited Visibility
}
\titlerunning{Point Convergence with Limited Visibility}
\relatedversion{A condensed  version of this paper appears in~\cite{KiKNPS21}.}



\author{David Kirkpatrick}{Dept. of Computer Science, University of British Columbia, Vancouver, Canada}
{kirk@cs.ubc.ca} {}{}

\author{Irina Kostitsyna}{Dept. of Mathematics and Computer Science, TU Eindhoven, Eindhoven, the Netherlands}{i.kostitsyna@tue.nl}{}{}

\author{Alfredo Navarra}{Dept. of Mathematics and Computer Science, University of Perugia, Perugia, Italy}{alfredo.navarra@unipg.it}{}{}

\author{Giuseppe Prencipe}{Dipartimento di Informatica, University of Pisa, Pisa, Italy}{giuseppe.prencipe@unipi.it}{}{}

\author{Nicola Santoro}{School of Computer Science, Carleton University, Ottawa, Canada}{santoro@scs.carleton.ca}{}{}

\authorrunning{D. Kirkpatrick, I. Kostitsyna, A. Navarra, G. Prencipe, and N. Santoro}

\begin{document}

\maketitle

\begin{abstract}

Among fundamental problems in the context of distributed computing by autonomous mobile entities, one of the most representative 
and well studied
is {\sc Point Convergence}: given an arbitrary initial configuration of identical entities, disposed in the Euclidean plane, 
move in such a way that, for all $\eps>0$, a configuration in which the separation between all entities is at most $\eps$ is eventually 
reached and maintained.

The problem has been previously studied in a variety of settings, including full visibility, exact measurements (like distances and angles), 
and synchronous activation of entities. Our study concerns the minimal assumptions under which entities, moving asynchronously 
with limited and unknown visibility range and subject to limited imprecision in measurements,  can be guaranteed to converge in this way.
We present an algorithm that solves {\sc Point Convergence}, for entities in the plane, in such a setting, provided the degree of asynchrony is bounded: while any one entity 
is active, any other entity can be activated at most $k$ times, for some arbitrarily large but fixed $k$. This provides a strong positive answer to 
a decade old open question posed by Katreniak.

We also prove that in a comparable setting that permits unbounded asynchrony, {\sc Point Convergence} in the plane is impossible, contingent on the natural 
assumption that algorithms maintain the (visible) connectivity among entities present in the initial configuration. This variant, 
that we call {\sc Cohesive Convergence}, serves to distinguish the power of bounded and unbounded asynchrony in the control of autonomous 
mobile entities, settling at the same time a long-standing question whether in the Euclidean plane synchronously scheduled entities are 
more powerful than asynchronously scheduled entities.
\end{abstract}

\Copyright{D. Kirkpatrick, I. Kostitsyna, A. Navarra, G. Prencipe, N. Santoro}

\section{Introduction}
\label{intro}

\subsection{Framework and background}
The notion of {\em distributed computational system},
initially restricted to systems of
processing units interacting by
exchanging  messages through  a
fixed communication network,
over time 
has been broadened  to include 
not only other forms of communication 
(e.g., via shared memory), but also very
different types of computational entities
and operational environments, from
mobile agents to programmable particles,
from mobile wireless ad-hoc networks to robot swarms.

The main theoretical
questions in all these environments are the same: what are
the minimal conditions that allow the entities to solve certain fundamental
problems?  What is the computational relationship between 
different conditions?
In this paper, we address these questions for the 
standard ${\mathcal OBLOT}$ model
of systems of autonomous mobile entities. 

In this model, an arbitrarily large collection of identical computational entities, called robots, operate in 
one or higher dimensional
Euclidean space within which they can move without constraint. 
Robots have no memory of past configurations or actions, no global knowledge about the system beyond what they are able to sense within some fixed but possibly unknown visibility range,
and no means of direct communication with other robots. 
In a single active cycle, a robot plans, following a fixed built-in algorithm informed by what it currently senses alone, and then initiates a change of position, that may stop short of being fully realized.

Since the only input
available to and the only output produced by  individual active cycles 
are geometric positions, computations  in this model are very different from those in
other distributed universes (e.g., message-passing networks).
Informally then, 
computations in this model are just movements of the robots from 
their initial configuration toward  other goal configurations,
and (computational) problems are typically expressed as  
temporal geometric 
predicates on goal and intermediate configurations.



For example,  simple {\em pattern formation} problems  require
the robots  to  reach,
through a sequence of intermediate configurations that maintain some geometric invariant 
(such as collision avoidance or the preservation of visibility),
 a configuration where their
 positions satisfy the geometric predicate (the {\em pattern}) of that 
 problem 
 and remain stationary thereafter.
In contrast, for  simple {\em pattern convergence} problems,  the robots are required
  to form the pattern in the limit.

 
An algorithm {\em solves} a problem
  if, under its control, {\em starting from any valid initial configuration}, the robots reach configurations that satisfy its temporal geometric predicate
 {\em in all valid activation sequences}.
The validity of activation sequences is determined by
the level of {\em synchronization}
of the system;
as in other distributed environments 
(e.g., message-passing networks),
the primary distinction is between 
synchronous and  asynchronous scheduling models.

In {\em synchronous} systems (\ssync), activity is divided into discrete 
intervals, called rounds; 
in each round, activated robots
perform their activity
cycle 
in perfect synchronization. 
The only constraint is {\em activation fairness}: every robot is activated infinitely often.
This very general
model, first studied in
\cite{SY99},  is also sometimes called {\em semi-synchronous}
to distinguish it from the {\em fully-synchronous}
model
(\fsync),
the special case where all 
robots
are active in every round.

In contrast, in {\em asynchronous} systems (\async), first studied in~\cite{FPSW99},
robots are activated 
at arbitrary times (constrained again by activation fairness), independently
of the other robots,  and the duration of individual active cycles, though finite,
is unpredictable.

There are many studies on the feasibility of problems 
and the complexity of algorithms
in both the synchronous and asynchronous settings,
e.g., see~\cite{Ando1999,CDN18,CDN19,CDN18a,CFPS12,CP05,DDFN18,FPSV17,FPSW08,FPSW05,CDNx20,SugS96,SY99,YaS10}.
In spite of the long and intensive research efforts, 
a fundamental question  is still outstanding:
{\em Is there any difference in computational scope between  algorithms
working within the \ssync and \async scheduling models?}
More precisely:  Does the assumption of synchronous scheduling of activity make possible the 
solution of some problem that is not solvable in the asynchronous setting?

 Any problem  providing  a positive  answer would 
computationally {\em separate} the two settings.
The quest for a separator has motivated a large number of studies,
but so far no  computational difference
between \ssync and \async has been
determined. 
In this paper, we provide an affirmative answer to the question. 
In fact, our results demonstrate an even stronger separation: between scheduling models with unbounded asynchrony (\async) and bounded asynchrony (what we call $k$-\async).


\subsection{Related work} 

Previous work on the computational power of collections of autonomous mobile robots has concentrated on the study of pattern formation and pattern convergence problems.\footnote{Many other problems can be reduced to pattern formation ones (e.g., leader election, gathering) or require pattern formation as an integral step (e.g. flocking).}
The vast majority of works assume that the visibility range of
the robots is {\em unlimited}.

\subsubsection{Pattern formation}


The first question studied has been what patterns can be formed 
from a given initial configuration~\cite{SY99}. 
Assuming that robots share a common handedness (also called chirality),
a sequence of results  has shown that the set of formable patterns is determined solely
by the  symmetry of the initial configurations, 
\cite{CDNx20,SY99,YaS10}.
In other words, with chirality, there is no distinction between the patterns  formable
from an initial configuration in  the \ssync and \async models.
No  similar general results are known without chirality.

Two specific pattern formation problems exhibit the level of symmetry 
necessary 
for their solution  starting from all initial configurations,
namely {\sc Point Formation} 
and {\sc Uniform Circle Formation}. 
{\sc Point Formation} corresponds to a basic coordination task in swarm robotics, called  {\sc Gathering} or {\sc Rendezvous}, where all robots are required to meet at a common point (not fixed in advance).

Remarkably, these two pattern formation problems 
have been solved in the most general \async without chirality in~\cite{CFPS12} and~\cite{FPSV17}, respectively; hence they cannot serve to separate the \async and \ssync models. 
This, in part, stimulated the study of the 
less stringent, but still well-motivated, pattern convergence variant of {\sc Point Formation},
namely {\sc Point Convergence}.  
The {\sc Point Convergence} problem 
requires robots to move in such a way that, for all $\varepsilon > 0$, a configuration in which the separation between all robots is at most $\varepsilon$.  
{\sc Point Convergence} is discussed in more detail below.
 
The converse question of determining from which initial configurations 
it is possible to form any possible pattern,  called {\sc Arbitrary Pattern Formation}, was
 posed and studied in~\cite{FPSW08} for \async robots without chirality; 
the answer to this question has been provided 
in~\cite{CDN19},
 where it has been shown that the set of 
possible initial configurations coincides with those from which it is possible to solve 
another basic problem,  namely {\sc Leader Election}. 

 Related investigations include  the {\sc Embedded Pattern Formation} problem, where the pattern is given to the robots as visible points in the plain, solved by \async robots without chirality~\cite{CDN18c};
and   the problem  {\sc  Pattern Sequence} of forming   a
 sequence of patterns in spite of  obliviousness, solved  by  \ssync robots with chirality~\cite{DaFSY15}.

When the visibility range of the robots is {\em limited},
 the setting considered in this paper, relatively little is known.
 It  has been recently shown that 
 limited visibility drastically  reduces  the set of formable patterns 
by  \async robots~\cite{YuY13}. 

\subsubsection{{\sc Point Convergence} under unlimited visibility}

Without further restrictions to the capabilities of the robots, 
{\sc Point Convergence} has been solved,
in~\cite{CP05}, by means of the so-called \texttt{Go\_To\_The\_Centre\_Of\_Gravity} (\texttt{CoG}), algorithm. 
In this algorithm robots move toward the center of gravity of the configuration, the unique point where the weighted position vectors of all the robots' locations sum up to zero. 
Note that the center of gravity may change as robots move. 
The convergence rate of this algorithm is $O(n^2)$, measured in terms of the number of rounds requited to halve  
the diameter of the convex-hull of the $n$ robots' positions (where a {\em round} corresponds to the minimal time during which all robots have completed at least one activation cycle). A lower bound of $\Omega(n)$ on the convergence required to solve {\sc Convergence} was also provided.

In~\cite{CDFHKK11}, a new algorithm was demonstrated that solves {\sc Point Convergence} in an asymptotically optimal $\Theta(n)$ 
number of rounds.
It applies the \texttt{Go\_To\_The\_Center\_Of\_Minbox} (GCM) strategy, where the \emph{minbox} is the minimal axes aligned box containing all current robots' positions. 
In fact, when robots share the notion about the orientation of the coordinate system axes, the algorithm requires constant time for halving the diameter of the convex hull. 

In~\cite{PMMS17,PMSS21}, two algorithms 
were developed, taking into consideration 
further relaxations of the underlying model. 
First, 
occlusion is taken into account,
that is, a robot $\bX$ looking at a robot $\bY$ cannot see anything behind $\bY$ (along the ray from $\bX$ that passes through $\bY$). Second, the proposed algorithms do not depend on the actual distances of the robots but two main variants are considered: (1) a robot can only deduce whether another robot is closer than a given constant distance $\delta$ or not; (2) an orthogonal line agreement model is assumed in which robots only agree on a pair of orthogonal lines. For (1) the proposed algorithm makes a robot moving toward other robots that are at distance greater than $\delta$ from itself. In (2), instead, still an algorithm based on the minbox strategy is applied.

Other interesting works on {\sc Point Convergence} concern the possible inaccuracy of the measurements (in terms of distances or angles) of the robots. 
The first investigation 
taking this kind of error into account can be found in~\cite{CP08}. In particular, the authors show that still the \texttt{CoG} algorithm can solve {\sc Convergence} but in the \fsync setting, whereas it fails in the \async setting. The problem remains open for the \ssync setting. Furthermore, they introduce a modification of the algorithm (called \texttt{Restricted\_CoG}) that ensures convergence in the \ssync setting and in the 1-dimensional \async setting. The problem was left open for the 2-dimensional case of \async robots.


An extension of~\cite{CP08} in terms of solvability under 
imprecise
measurements can be found in~\cite{YIKIW12}. 
The considered errors are restricted to uniform behaviours, that is, the same percentage of error with respect to the measurements is assumed for all the distances and angles perceived during an observation by a robot. In such a uniform case, it is shown that the convergence problem can be solved in a significantly wider set of cases.  \\

\subsubsection{{\sc Point Convergence} under limited visibility}
Very few investigations have focused on the {\em limited visibility} setting,
and among these none considers measurement errors (see~\cite{BaM19}, for a recent survey).

A central issue that arises when dealing with limited visibility concerns \emph{connectivity}. 
When two robots 
are within distance $V$, the visibility range which may or may not be known to the robots, they are assumed to be visible to each other.
Any instantaneous configuration of robots induces a \emph{visibility graph} that is usually (and naturally) assumed to be connected. In that case the configuration itself is said to be {\em connected}.



In~\cite{FPSW05} it has been shown that
{\sc Point Formation}
is solvable in \async if the robots are endowed with global consistency of the local coordinate system.
A variant of {\sc Point Convergence} is the {\sc Collisionless Convergence} problem,
which requires the robots to  converge without two (or more) robots ever occupying the same point at the same time. This has been studied in~\cite{PPV15},
showing it can be solved in \async when all the robots are assumed to have a compass (hence they agree on the ``North'' direction), but they do not necessarily have the same handedness (hence they may disagree on the ``West'').


The pioneering work on {\sc Point Convergence} under limited visibility is
by Ando et al.~\cite{Ando1999}, in the \ssync scheduling model.
The authors propose and analyze an algorithm that leads to convergence, by means of \emph{cautious} moves, starting from any connected configuration of robots.
In their \texttt{Go\_To\_The\_Centre\_Of\_The\_SEC} algorithm,
each active robot takes a snapshot of its surrounding (up to a known limited visibility range $V$), computes the center of the \emph{Smallest Enclosing Circle} (SEC) of the perceived robots,\footnote{The smallest circle that encloses all the robots perceived during the \look phase at distance not greater than $V$.} and then makes a move toward the center of the SEC for a distance that guarantees to preserve the visibility of these perceived robots, regardless of the movements of other robots following the same protocol. 
(Quite recently, Braun et al.~\cite{BCF20} considered
an extension of the \texttt{Go\_To\_The\_Centre\_Of\_The\_SEC} algorithm 
applied in the three dimensional Euclidean space, assuming a 
\fsync and continuous time scheduling model.)

In subsequent work~\cite{K11}, Katreniak introduced the
more powerful $k$-\async  
scheduling model, that permits robot activations with bounded asynchrony,
and presents a solution algorithm for {\sc Point Convergence} in the 1-\async model (a solution under 
a more restricted version of the 1-\async scheduler was presented earlier in~\cite{LMA07}).
The resolution of {\sc Convergence} in the $k$-\async model, for $k>1$, was left as an open problem.


The algorithms of Ando et al. and Katreniak,
are reviewed in more detail in Section~\ref{sec:NewAlgorithm}.
It is worth pointing out here that these, like all known algorithms that deal with limited visibility, 
have the property that
if two robots are initially within visibility range of each other,
they remain so at all times. 
Thus they solve a more restricted variant of  {\sc Point Convergence}, that we  call 
 {\sc Cohesive Convergence},  that includes this property as a invariant.  
It follows that, when solving  {\sc Cohesive Convergence}, an initially connected robot configuration will remain connected indefinitely. 
%

%

\subsection{Main contributions}

The focus of this paper,
as in  those of Ando et al.~\cite{Ando1999} and Katreniak~\cite{K11}, concerns the {\sc Cohesive Convergence} problem for identical autonomous robots with bounded visibility.
This requires any collection of such robots, starting from an arbitrary connected configuration, to move, under the control of a possibly adversarial scheduler, in such a way that, for all $\eps>0$, the robots are guaranteed to reach and maintain a configuration in which all robots lie within a circle of radius at most $\eps$. 
Our primary contributions are threefold: (i) a novel moderately error-tolerant algorithm that succeeds in the $k$-\async scheduling model for any fixed $k$, an environment that permits scheduling of robot activity with an arbitrarily large but bounded degree of asynchrony, (ii) an analysis that allows us to go well beyond what is possible with previous techniques that depend on the assumption of a significantly more restricted scheduling environment (\ssync or $1$-\async), shedding new light on the power of bounded asynchrony, and (iii) a family of robot configurations that demonstrates the impossibility of solving {\sc Cohesive Convergence} in \async, the fully asynchronous scheduling environment, provided algorithms are required to tolerate a modest amount of imprecision in perception (significantly less than what can be tolerated by our $k$-\async algorithm).

Even in the same scheduling environment, our algorithm (described in its basic form in Section~\ref{sec:NewAlgorithm}, with extensions, including error-tolerance, and higher-dimensional generalizations, outlined in Section~\ref{sec:Extensions}) is slightly simpler in formulation and provides a more general solution than its predecessors developed in~\cite{Ando1999,K11,LMA07}. 
In particular, the motion function of a robot depends only on the directions to the pair of visible neighbours, among those that are relatively distant, that define the maximum angular range. Furthermore,  distance and angle measurements to visible neighbours are assumed to be accurate only to within some limited imprecision, and robot motion is subject to some limited error in both distance and angle. 
Finally, our algorithm exhibits a degree of uniformity not evident in its predecessors: specifically, (a) dependence  on the visibility range $V$ is not built into the algorithm, and (b) our algorithm is simply formulated to work in an environment with a base level of asynchrony; to function in an environment that permits a higher degree of asynchrony, captured by a parameter $k$, one can simply scale the motion function used by our algorithm in its base formulation by a factor of $1/k$.

As with the analysis of algorithms in~\cite{Ando1999,K11,LMA07}, our proof of convergence involves (i) an argument that the connectivity of the initial configuration is preserved, and (ii) an argument that congregation, i.e. convergence to an arbitrarily small region, eventually takes place.  
Our connectivity preservation analysis (presented in Section~\ref{sec:VisibilityPreservation}) is complicated by the fact that, harnessed by an adversarial scheduler, even the most basic level of asynchrony makes it impossible for a robot $\bX$ to accurately discern the influence of its own location on the motion of an active neighbour $\bY$. 
This is further compounded by the fact that in environments with a higher degree of asynchrony $\bY$ might have seen $\bX$ many times during its current motion. 
In order to help isolate the most salient features of our analysis in asynchronous environments, we first introduce a restriction on the $k$-\async environment, requiring intersecting activity intervals to nest (one entirely within the other). By staging our analysis in this way, we provide a clearer picture of the extent to which the difficulties associated with asynchrony in our setting arise from (possibly lengthy) chained, as opposed to nested, activations. 
To demonstrate the preservation of visibility under arbitrarily long chains of activations we introduce a novel
backward reachability analysis that fully exploits the relative simplicity of our algorithm’s motion function. 
For our congregation argument (presented in Section~\ref{sec:IncrementalConvergence}), we start by observing, as was the case in~\cite{Ando1999,K11,LMA07}, that convergence to a convex configuration with bounded diameter follows directly from the fact that the convex hull of successive configurations are properly nested. 
Our argument that the limiting diameter is zero differs from, and is arguably simpler than, that used in these earlier works. 
The argument exploits a kind of hereditary property that ensures that, after some point in time, robot locations become relatively stable in the sense that once a robot has vacated a small neighbourhood of an extreme point of the configuration it must remain outside.

Demonstrating the impossibility of solving {\sc Cohesive Convergence} in any environment that assumes exact measurement and control seems daunting, since such assumptions open up the possibility of encoding information, perhaps history, in the location of robots. 
However, we show (in Section~\ref{sec:ImpossibilityA}) that algorithms that control more realistic robots that operate with even a very modest amount of imprecision in measurements in a scheduling environment with no bound on asynchrony (other that a fairness condition that ensures that no robot is permanently locked out of activity), cannot solve {\sc Cohesive Convergence}. 
This is demonstrated by a family of configurations with the property that for any algorithm there is a configuration in the family that can be forced to become disconnected into two linearly separable components. 
Interestingly, the adversarial scheduler that achieves this uses only nested activations (of necessarily unbounded depth). 

In summary, we prove that, considering algorithms that tolerate a very modest amount of measurement imprecision,  
 {\sc Cohesive Convergence} is {\em solvable}
in the $k$-\async scheduling model,  for any fixed $k$,  providing a strong positive answer to a question left open in~\cite{K11}. 
However, under even slightly weaker assumptions, the same problem is shown to be 
 {\em unsolvable} in the \async scheduling model.
In this sense,  {\sc Cohesive Convergence} provides a separation between scheduling models with bounded and unbounded asynchrony,
and ipso facto between the \ssync and \async scheduling models.
Our work raises several interesting questions, suitable for future research. Several of these are highlighted in our concluding section (Section~\ref{sec:Conclusion}). 


\section{Model}
\label{sec:model}

We consider the standard  ${\mathcal OBLOT}$ model of 
distributed systems of  mobile  entities (e.g.,  see~\cite{FPS12}). 
The system is composed of a set  
$\mathcal{R}=\{ \bX^1,\ldots ,\bX^n \}$ of $n \geq 1$  
oblivious computational  entities, called {\em robots}, 
 that reside  in  $d$-dimensional  Euclidean space
${\mathbb R}^d$,  $d\geq 1$, in which they can move,
 and operate in \look-\compute-\move activity cycles.\footnote{If $P$ and $Q$ are points in 
$\mathbb{R}^d$, we denote by $\overline{PQ}$ the line segment joining $P$ and $Q$, by $|PQ|$ the length of this segment, and by $\overrightarrow{PQ}$ the ray from $P$ to $Q$.}

\subsection{Robots}

We treat robots as dimensionless 
entities in , 
with $\mathbb{R}^2$; a discussion of extensions to higher dimensional spaces is left to Section~\ref{sec:Extensions}.  We use $X$ to denote the generic location of robot $\bX$; when needed $X(t)$ is used to denote the location of $\bX$ at a specific time $t$.
The 
multiset 
${\mathcal C}(t) = \{X(t): X\in {\mathcal R}\}$ specifying the 
locations of the robots at time $t$
is called the {\em configuration} of the system at time $t$.

Robots are equipped with sensorial devices that allow them
to observe the positions
of the other robots within 
a  fixed, but possibly 
unknown 
range $V$. 
More precisely, the  {\em visible region} of robot
$\bX$ at time $t$ is the  closed point set
{\em Vis}$(\bX,t)=\{Y : |X(t) Y| \le V\}$.
We say that a robot $\bY$  is 
a \emph{neighbour} of robot $\bX$ at time $t$ 
if $Y(t)$  belongs to {\em Vis}$(\bX,t)$.
%
The  associated 
{\em visibility graph}
 at time $t$ is
the undirected graph $G(t) = ({\mathcal R}, E(t))$, 
where $(\bX,\bY) \in E(t)$ if and only if 
$|X(t) Y(t)| \le V$.
A configuration of robots is said to be {\em connected} when its associated visibility graph is connected.



Robots 
are provided with a persistent read-only memory containing
a control algorithm 
as well as
a local volatile memory  for its computations, that could involve exact real arithmetic.
Robots are endowed with unrestricted motorial capabilities; 
that is, they are free to move in any direction.

Robots are  (i) {\em  autonomous}: that is,  they operate  without a central control or external supervision, 
(ii) {\em identical}:
that is,   they are  indistinguishable by their appearance, 
they do not have  
distinct identities that can be used during the computation, and
they all have and execute the same  
algorithm; and
(iii) {\em silent}: they have no means of direct
communication of information to other robots,
so any communication occurs in a totally implicit manner, by
moving and by observing the  positions of the robots within visibility range.

\subsection{Activity cycles and obliviousness}


The {behaviour} of each robot can be described as the alternation  
of finite length \emph{activity} intervals with  finite length \emph{inactivity} intervals. 
Each active interval consists of a sequence of three phases: \look, \compute, and \move, 
called  a \look-\compute-\move (\LCM) activity cycle. 
The operations performed by each robot $\bX$ in each phase are
as follows.

\begin{enumerate}
\item  \look. At the start of an activity interval, 
the robot observes the world within its visibility range.
The result is  an instantaneous {\em snapshot} 
of the positions\footnote{If the robots are endowed 
with multiplicity detection,  the snapshot  indicates also
whether there are two or more robots co-located at a point; otherwise such multiplicities are perceived as a single robot.}  
occupied by visible robots at that time;
these positions are expressed within a
local (i.e., private) 
coordinate system. 
The private coordinate systems 
of different robots at the same time, or the same robot at different times, need not be consistent; hence from a global point of view, 
the robots are {\em disoriented}.

\item  \compute. Using the snapshot as an input, 
the robot  executes  its built-in algorithm  ${\mathcal A}$,
the same for all robots, which determines
 an intended {\em destination} point.

\item  \move. If the destination point coincides with its current location,
$\bX$ is said to perform the \emph{nil} movement;
otherwise it moves toward the computed destination along a straight trajectory. 
The \move phase (and with it, the activity cycle) of a robot might end before it reaches its
destination.
\end{enumerate}

\noindent The duration of each \compute and \move phase
is assumed to be finite,\footnote{Note that 
also the \compute phase could be assumed to be instantaneous, without any loss of generality.} while the \look phase is assumed to be instantaneous. During a \move phase a robot is said to be {\em motile} (not necessarily moving, but capable of moving), otherwise it is {\em immotile}.



The local memory of the robots is volatile: at the end of a cycle, all contents are erased.
In other words, 
the robots are {\em oblivious}: at the beginning of each activity cycle, 
a robot has no memory of past actions and computations, 
and the computation is based solely on what determined in the current cycle.

\subsection{Adversarial control and conditions}

In the environment in which the robots operate,
there are several factors and conditions, in addition to the initial configuration, that are beyond the control of their algorithm but nevertheless 
impact their actions and, ultimately, their worst-case behaviour.  
These factors, which are assumed to be  under  adversarial control,
 are discussed in the following.

\subsubsection{Activation and Synchronization}

The mapping of robot activities to time, including activation/deactivation times, as well as the duration of \compute and \move phases within each activity cycle, and the rate of motion within each \move phase,
is determined by 
a
{\em scheduler} that is constrained only by the understood 
level (or model) of system synchronization.
In all models, the scheduler
must respect {\em activation fairness}: 
for each robot $\bX$ and each time $t$, there exists a time $t'>t$ where $\bX$ will be active.

 The main models are the synchronous   and the asynchronous ones: 

\begin{itemize}
\item
In the {\em synchronous} (also called {\em semi-synchronous}) model (\ssync)
time is logically divided into a sequence of   \emph{rounds}. 
In each round,  a {\em subset} of the robots is activated, and they 
perform each phase of their  \LCM cycle simultaneously, 
terminating their activity interval by 
the end of the round. 
The choice of which robots are activated in a round is made by the
scheduler.
A special case of  \ssync is the {\em fully-synchronous} model (\fsync)
where   all robots are activated in every round.

\item In the {\em asynchronous} model (\async), each robot is activated at arbitrary times, independently of the others;  in each activation interval, the duration of each \compute and \move phase is finite but unpredictable, and might be different in
different cycles.
The duration of the  \compute and \move phases of each cycle as well as the decision of when to
 activate a robot is controlled by the 
 scheduler.

\end{itemize}

\noindent An  illustration of the difference between the activations and durations
in \fsync, \ssync, and \async is shown in  Figure~\ref{newfig:models-a}.

\begin{figure}[t]
\centering
\includegraphics[page=1]{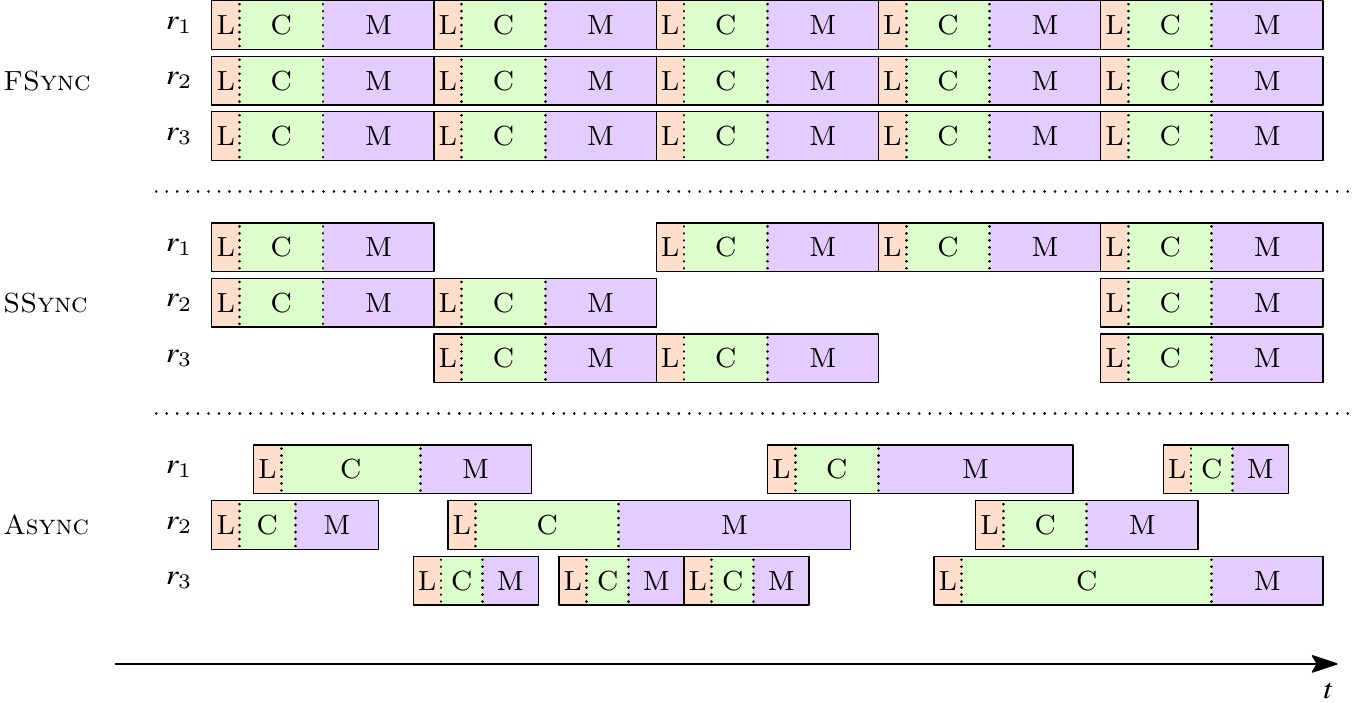}
\caption{ The execution model of active intervals for each of \fsync, \ssync, and \async  robots. Inactive intervals are implicitly represented by empty time periods.}\label{newfig:models-a}
\end{figure}

\begin{figure}[htbp]
\centering
\includegraphics[page=2]{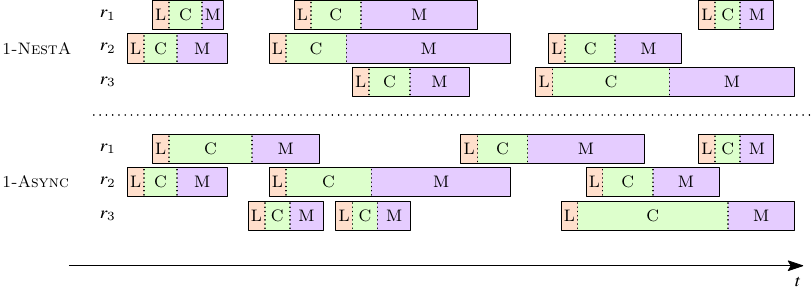}
\caption{The execution model of active intervals for $1$-\nesta and $1$-\async robots. }
\label{newfig:NewModels}
\end{figure}

The asynchronous model licenses an adversarial scheduler to depart from the synchronous model in two ways: (i) activation intervals of different robots can overlap arbitrarily; and (ii) one robot can be activated arbitrarily many times during a single activation of another robot. 
With this in mind, we first define the 
{\em nested-activation} model (\nesta) that restricts the activity intervals of all pairs of robots to be either 
disjoint or nested (no active interval  contains exactly one endpoint of another active interval).
We  will focus more precisely on  two natural restrictions of the \nesta and \async models, that lie between the \ssync and \async models (see Figure \ref{newfig:NewModels}):

\begin{itemize}
\item \emph{$k$-Nested-Activations} ($k$-\nesta): 
Activity intervals of any pair of robots are either disjoint or nested, and in addition, at most $k$ activity intervals of one robot can be nested within a single activity interval of another.
\item \emph{$k$-Asynchronous} ($k$-\async): As in \async, robots are activated independently and the duration of each activity interval is arbitrary (but finite). However, at most $k$ activations  of one robot can occur within a single active interval of another.
\end{itemize}



\subsubsection{Rigidity of motion}

In the \move phase, there is a constant $\xi \in (0, 1]$ (unknown to the robots, and possibly related to the visibility range $V$) such that robots are guaranteed to move at least fraction $\xi$ of the way toward their intended destination. The associated motion is said to be {\em $\xi$-rigid}. (If $\xi = 1$, motion is simply said to be {\em rigid}).%
\footnote{In the literature, non-rigid motion is usually qualified by an assumption that at least some non-trivial distance is traversed. In our setting, where it makes no sense for an algorithm to specify a motion of length that significantly exceeds the visibility range, this assumption can be subsumed by the $\xi$-rigid assumption, for suitable $\xi$.}



\subsubsection{Measurement imprecision and coordinate system distortion}

In error-tolerant settings,  
the accuracy of the measurements (distances and angles) made by a robot during its \compute phase, as well as
that of 
a robot's ability to realize its intended trajectory 
during its \move phase,
are considered to be subject to adversarial control. 

One natural way in which relative error in angle measurements could arise is from a small symmetric distortion of a robot's local coordinate system: a continuous bijection $\mu: [0, 2\pi]) \rightarrow [0, 2\pi)$ with $\mu(\theta + \pi) = \mu(\theta) + \pi$, for all $\theta \in [0, \pi)$.  
    Such a distortion would result in a robot's local coordinate system deviating slightly from any rigid transformation of an unknown global coordinate system, yet still permit consistent orientation with respect to neighbours.


\subsection{Problems and solutions}

 Since  robots 
 can only observe the positions of others and  move,
 in the ${\mathcal OBLOT}$ model  a {\em problem} to be solved (or, a task
 to be performed) is expressed in terms of  a   {\em temporal geometric predicate},
 which  the configurations formed by the robots' positions
from some time on  must satisfy.

The {\sc Point Convergence} problem 
requires the robots, starting from an arbitrary connected configuration
\footnote{It is possible to define Convergence in a meaningful way that applies to unconnected initial configurations as well. We will discuss this further in Section~\ref{sec:Extensions}.}
where they are all inactive,
 to become arbitrarily close. 
 %
 More precisely, {\sc Point Convergence} is the problem
 defined by the temporal geometric predicate: 
 \[
 {\tt Convergence} \equiv
\forall \eps\in {\mathbb R}^+, 
\exists t:  \forall t'\geq t, \forall \bX,\bY\in\mathcal{R},  
|X(t') Y(t')|\leq\eps\,.
\]

A more constrained version is the
the {\sc Cohesive Convergence} problem that additionally requires
that, 
if two robots initially are within the  visibility region
of each other, they remain so at all times;
 i.e.  it is defined by the temporal geometric predicate:
\[
 {\tt CohesiveConvergence} \equiv
{\tt Convergence} \wedge  (\forall t\geq 0,  E(0)\subseteq E(t))\,.
\]
Note that even though {\sc Cohesive Convergence} is strictly more constrained than {\sc Point Convergence}: all known algorithms for {\sc Point Convergence} also solve {\sc Cohesive Convergence}.


 An algorithm 
 ${\mathcal A}$  is said to {\em solve} the {\sc Point Convergence} (or {\sc Cohesive Convergence}) problem if it satisfies the corresponding predicate, starting from \emph{any} valid initial configuration of robots, under \emph{any} 
 (possibly adversarial) scheduler that respects the associated synchronization model.  
 If every \move is chosen in such a way that it cannot transform a static connected configuration into a disconnected configuration, the algorithm is said to be \emph{coherent}.

\section{A new algorithm for convergence}\label{sec:NewAlgorithm}


We begin this section with a review of the algorithms of~\cite{Ando1999} and~\cite{K11}, emphasizing the features common to our algorithm, and features on which our algorithm differs. This is followed by a detailed overview of our new approach.

\subsection{The convergence algorithms of Ando et al. and Katreniak.}
%
%
The algorithm of Ando et al.~\cite{Ando1999} is the following: upon activation, each robot $\bX$
\begin{itemize}
   \item locates, within its local coordinate system, all of the other robots within its visibility range;
    \item computes a \emph{safe region} for motion with respect to each of its  neighbours, and
     \item computes 
     the center of a minimum enclosing ball  of all the robots within its visibility range, and
    moves as far as possible towards 
    this center
    while remaining inside all of the individual safe regions.
\end{itemize}

The algorithm of Ando et al. is formulated in the \ssync scheduling model. It assumes:
\begin{itemize}
    \item robots are points that execute their full planned motions instantaneously, i.e. rigid collision-free motion;%
    \footnote{The assumption of rigid motion, though standard at the time of~\cite{Ando1999}, could be relaxed without impacting the correctness of the algorithm.}
    \item knowledge of the common visibility radius $V$ is built into the algorithm;%
    \footnote{Again, this assumption is stronger than necessary: it would suffice to replace $V$ by the distance to the furthest visible neighbour.}
    \item there is no error in perception or in the realization of planned motion.
\end{itemize}

The algorithm of Katreniak~\cite{K11}  has many similarities to that of~\cite{Ando1999}: upon activation, each robot $\bX$
\begin{itemize}
   \item locates, within its local coordinate system, all of the other robots within its visibility range, 
    \item computes a \emph{safe region} for motion with respect to each of its  neighbours, 
    and
     \item moves as far as possible while remaining inside a composite safe region that respects all of the individual safe regions.
\end{itemize}

Katreniak's algorithm is formulated in the $1$-\async model. It assumes:
\begin{itemize}
    \item knowledge of a common visibility radius $V$ is not assumed; instead each robot $\mathbb Z$ works with a (changing) lower bound $V_{\mathbb Z}$ determined by the distance to its furthest visible neighbour;
    \item there is no error in perception;
    \item the planned motion might not be fully realized, subject to a natural {\em progress} condition, that ensures eventual convergence.
\end{itemize}

The most apparent difference between the algorithms of Ando et al. and Katreniak lies in the specification of their safe regions. 
For a robot $\mathbb Y$ located at $Y_0$ viewing a robot $\mathbb X$ located at $X_0$,
Ando et al. specify a safe region as a disk with radius $V/2$ centred at the midpoint between $X_0$ and $Y_0$.
Katreniak's safe region is formed by the union of two disks, one with radius $|X_0Y_0|/4$ centred at the point $(X_0 + 3 Y_0)/4$, and the other with radius $(V_{\mathbb Y}-|X_0Y_0|)/4$ centred at $Y_0$ 
(see Figure~\ref{fig:SafeRegionA+K+New}).

\begin{figure}[t]
\centering
\includegraphics[page=1]{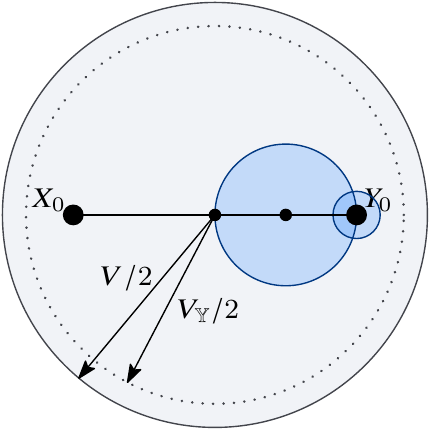}
\hfill
\includegraphics[page=2]{newfig/safe-region.pdf}
\hfill
\includegraphics[page=3]{newfig/safe-region.pdf}
\caption{Basic safe region for motion of robot $\bY$ at location $Y_0$ (with respect to visible robot $\bX$ at location $X_0$), as specified by Ando et al. (grey), by Katreniak (blue), when $\bX$ is a distant neighbour of $\bY$ (left) and when it is a close neighbour (center); basic safe region in our scheme (green), assuming $\bX$ is a distant neighbour of $\bY$ (right).}
\label{fig:SafeRegionA+K+New}
\end{figure}

The correctness of 
the convergence algorithms in both 
\cite{Ando1999} and~\cite{K11} 
rest on two observations: 
(i)~\emph{(initial) visibility preservation}: the choice of safe region guarantees that all robot pairs that are mutually visible in the initial configuration will remain mutually visible thereafter,
and (ii)~\emph{incremental congregation}: the trajectories of robots following the 
algorithm
exhibit a notion of ``progress'' towards convergence (measured in terms of the shrinking  of the convex hull of the robot locations). In both cases, the argument is complicated by the need to demonstrate that the limit of convergence is a single point.

\begin{figure}[tbp]
\centering
\includegraphics{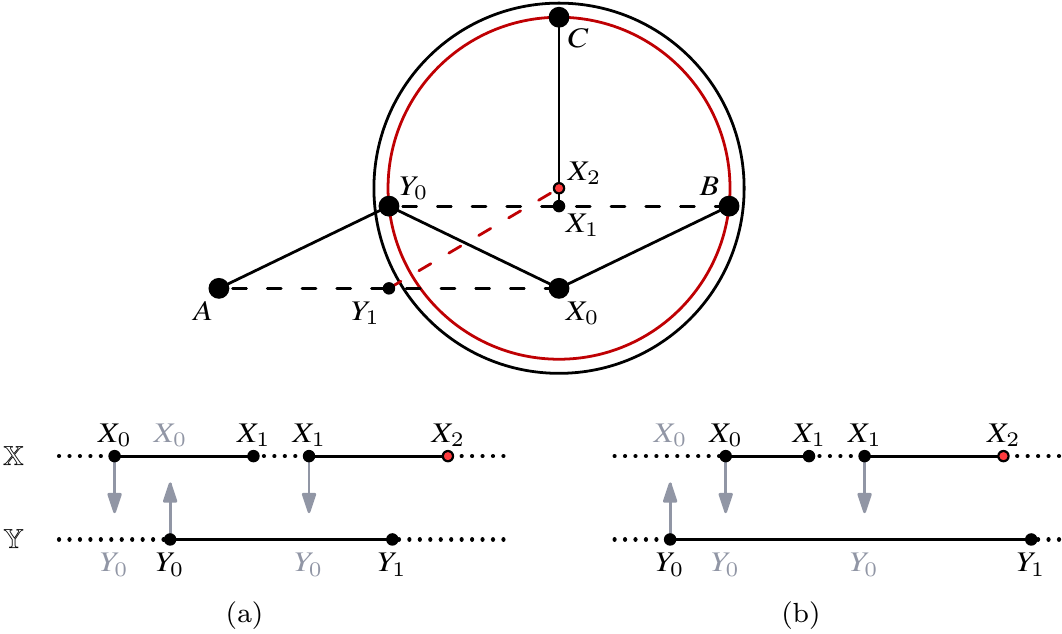}
\caption{An example where the (unmodified) Ando algorithm leads to separation in the (a)~1\nobreakdash-\async model and (b)~2-\nesta model. The example consists of five robots: $\mathbb{A}$, $\mathbb{B}$, and $\mathbb{C}$ being stationary through the example (the scheduler does not activate these robots); and $\bX$ and $\bY$ being activated according to the two activation timelines on the bottom. Horizontal axis represents the time, grey arrows designate the snapshots, and the corresponding positions of the robots $\bX$ and $\bY$ are also shown in grey. The red circle, centered at $X_2$, is the smallest enclosing circle of the points $X_0$, $Y_0$, $B$, and $C$. The black circle is centered at $X_2$ and has radius $V$. In both cases, robot $\bX$ moves from location $X_0$ to location $X_2$, and robot $\bY$ moves from $Y_0$ to $Y_1$, leading to the distance between the two robots being strictly greater than $V$.}
\label{newfig:AndoFailure}
\end{figure}

Note that, even when $k=1$, both the  $k$-\nesta and the $k$-\async models provide modest generalizations of the \ssync model. 
Specifically,
by definition, any algorithm that guarantees convergence in the $1$-\nesta (or $1$-\async) model will do the same in the \ssync model.
However,
\begin{enumerate}[(i)]
\item In general, Ando's algorithm (unmodified) does not succeed in the $1$-\async model, even if it is able to assume that moves are rigid and instantaneous 
(so the scheduler cannot pause or stop robots at intermediate points in their planned trajectory);\footnote{Note that Katreniak makes a similar observation, using an example that exploits non-rigid motion.}
see Figure~\ref{newfig:AndoFailure} (a);
\item The same construction, with a slightly modified timeline (Figure~\ref{newfig:AndoFailure} (b)), shows that Ando's algorithm (unmodified) does not succeed in the $2$-\nesta model;
\item It is also not difficult to show that Katreniak's algorithm (unmodified) does not succeed in the $k$-\async scheduling model, when $k$ is sufficiently large.
\end{enumerate}


\subsection{Overview of our algorithm.}
%
In our algorithm, as in earlier schemes, each robot $\bY$ starts an active interval by 
locating, within its local coordinate system, all of the other robots within its visibility range, referred to as \emph{neighbours}. 
It is not assumed that the visibility radius $V$ is known. Instead, in each active interval, robot $\bY$ computes a (tentative) lower bound $V_{\bY}$ on  $V$, provided by the distance to its current furthest neighbour. Neighbours whose distance is greater than $V_{\bY}/2$ are referred to as \emph{distant} neighbours of $\bY$; others are  \emph{close} neighbours. Note that close and distant are relative and robot-specific notions; in particular, $\bY$, by definition, always has at least one distant neighbour. Furthermore, $\bX$ could be a close neighbour of $\bY$ while $\bY$ is a distant neighbour of $\bX$.


Robot $\bY$ continues by 
determining a safe region for motion with respect to each of its neighbours $\bX$.
Safe regions are designed to ensure that, despite the fact that $V$ is unknown, the connectivity of the 
initial robot configuration
will not be lost; in particular, if robots $\bX$ and $\bY$ are mutually visible in their initial configuration (i.e. $|X(0) Y(0)| \le V$), 
when both are assumed to be immotile then, provided both robots 
continue to confine their movement to the safe regions for motion with respect to the other,
their mutual visibility will be maintained thereafter. 
Even though acquired visibility might be subsequently lost, unlike the schemes of~\cite{Ando1999} and~\cite{K11}, a form of acquired visibility preservation can be shown to hold, which suffices to ensure that there is a point in time after which  if robots $\bX$ and $\bY$ are even momentarily not mutually visible, then $\bX$ and $\bY$ remain separated by distance at least $V/4$. 
The specification of safe regions, together with their properties that are exploited in subsequent visibility preservation and congregation arguments, is provided in the next subsection; the details of our visibility preservation arguments in models with bounded asynchrony, are developed in Section~\ref{sec:VisibilityPreservation}.

Finally, $\bY$ plans a motion to a target destination contained within the intersection of the safe regions associated with all of its neighbours. 
The target destination in turn is designed to ensure that the convex hull of the robot locations shrinks monotonically, and converges to a point. The specification of the target destination, together with details of this congregation argument, are presented in Section~\ref{sec:IncrementalConvergence}.


\subsubsection{Safe regions: specification and properties.}
The \emph{basic safe region} for motion of robot $\bY$, located at $Y_0$, with respect to a distant neighbour $\bX$ at location $X_0$, denoted $S^{V_{\bY}/8}_{Y_0}(X_0)$, is a disk,
with radius $V_{\bY}/8$, centered at the point at distance $V_{\bY}/8$ from $Y_0$ in the direction of $X_0$
(cf. Figure~\ref{fig:SafeRegionA+K+New}).

Note that (i)~in order to respect the basic safe region of any distant neighbour, of which there must be at least one, a robot can plan a motion of length at most $V_{\bY}/4$. 
(In fact, as will become clear, when the actual destination point, respecting the safe regions of all visible neighbours, is specified, a robot will never plan a motion of length greater than $V/8$.)
Hence, the planned motions of close neighbours cannot possibly lead directly to a separation exceeding $V$, 
(ii)~unlike the algorithms in~\cite{Ando1999} and~\cite{K11}, our basic safe regions, defined for distant neighbours only, depend only on the direction of that neighbour,
and (iii) if a robot $\bX$ is located in the convex hull of its distant neighbours then the intersection of the safe regions with respect to the distant neighbours contains its current location only, so to respect all of these safe regions it must remain stationary.\footnote{It is worth noting that the choice of $V_{\bY}/2$ as the definition of ``close'' is somewhat arbitrary. 
Similarly, the size of the safe region associated with distant neighbours has been chosen to have radius $V_{\bY}/8$ mostly for convenience. Anything less than this would certainly work as long as it is at least some positive constant. (Furthermore, choosing a smaller radius would allow us to choose a larger radius for ``close'').}
%
%
In the $1$-\nesta and $1$-\async models robots are constrained to choose a target location that lies within the basic safe region with respect to each of their neighbours. In the $k$-\nesta and $k$-\async models, with 
$k >1$, the only change is to simply scale the basic safe regions by a factor of $\alpha = 1/k$. 
Thus for all $k$,  robot $\bY$ is constrained to choose a target location that lies within the $\alpha$-scaled safe region with respect to each distant neighbour robot $\bX$, denoted $S_{Y_0}^{\alpha V_{\bY}/8}(X_0)$, which is
the disk, with
radius $\alpha V_{\bY}/8$, centered at the point at distance $\alpha V_{\bY}/8$ from $Y_0$ in the direction of $X_0$.
It is straightforward to confirm that if point $P$ lies within the basic safe region $S^{V_{\bY}/8}_{Y_0}(X_0)$, then the point $P^\alpha$ at distance $\alpha |P Y_0|$ from $Y_0$ in the direction of $P$, lies in $S_{Y_0}^{\alpha V_{\bY}/8}(X_0)$.

The prospect of one robot, $\bY$, at location $Y_0$, making up to $k$ successive moves while another distant neighbour, $\bX$, is in the process of moving from location $X_0$ to location $X_1$, invites the question: how can we characterize the  set of points that can be reached by robot $\bY$ in this situation, provided that each of its moves is confined to the current $1/k$-scaled safe region with respect to the current location of $\bX$? As it happens the exact description of this set is somewhat complicated, but a more simply described superset suffices for our purposes. 

\begin{figure}[htbp]
	\centering
	\includegraphics[page=1]{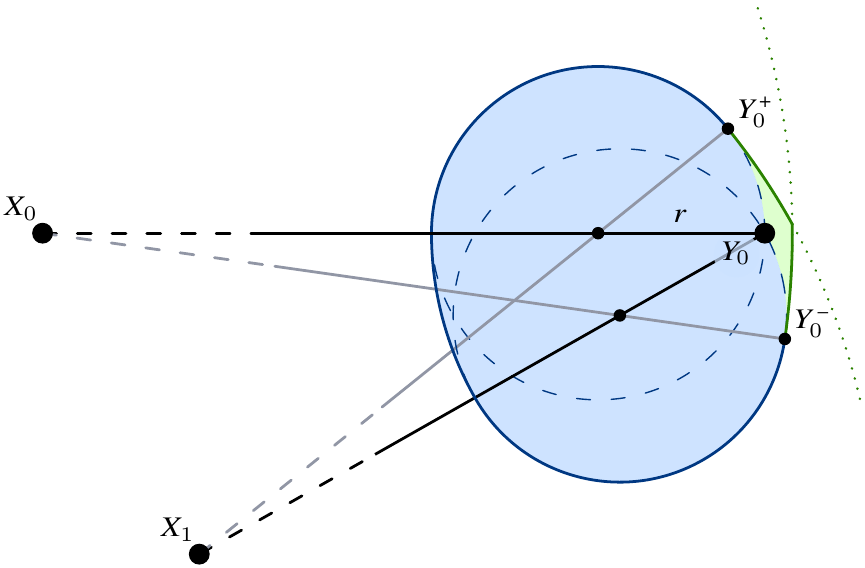}
	\caption{Region $R^r_{Y_0}(X_0, X_1)$ consisting core (blue) and bulge (green) subregions.}
	\label{newfig:QuarterNewProofBaseRegion}
\end{figure}

We define the region $R^r_{Y_0}(X_0, X_1)$ 
(see Figure~\ref{newfig:QuarterNewProofBaseRegion}) as the union of
\begin{enumerate}[(i)]
    \item its \emph{core}: the set of disks of radius $r$ whose centers lie at distance $r$ from $Y_0$ in the direction of some point on the line segment $\overline{X_0X_1}$; and
    \item its \emph{bulge}: the intersection of 
(a) the set of points at distance at most $|X_1 Y_0^+|$ from $X_1$, and at most $|Y_0 Y_0^+|$ from $Y_0$, where $Y_0^+$ is the point on the disk centered at distance $r$ from $Y_0$ in the direction of $X_0$ that has maximum distance from $X_1$; and
(b) the set of points at distance at most $|X_0 Y_0^-|$ from $X_0$, and at most $|Y_0 Y_0^-|$ from $Y_0$, where $Y_0^-$ is the point on the disk centered at distance $r$ from $Y_0$ in the direction of $X_1$ that has maximum distance from $X_0$.
\end{enumerate}

%

We begin by making the following observations that follow directly from the definitions above:
\begin{obs}~
\begin{enumerate}[(i)]
\item $R^{V_{\bY}/8}_{Y_0}(X_0, X_0)$ coincides with $S^{V_{\bY}/8}_{Y_0}(X_0)$ 
and, more generally, for any $\alpha \in (0,1]$, \\
$R^{\alpha V_{\bY}/8}_{Y_0}(X_0, X_0)$ coincides with $S^{\alpha V_{\bY}/8}_{Y_0}(X_0)$; and
\item the core of $R^{V_{\bY}/8}_{Y_0}(X_0, X_1)$ coincides with the union of $S^{V_{\bY}/8}_{Y_0}(X_*)$, over all  $X_* \in \overline{X_0 X_1}$,
and, more generally, for any $\alpha \in (0,1]$,
the core of $R^{\alpha V_{\bY}/8}_{Y_0}(X_0, X_1)$ coincides with the union of $S^{\alpha V_{\bY}/8}_{Y_0}(X_*)$, over all  $X_* \in \overline{X_0 X_1}$.
\end{enumerate}
\end{obs}
%
\noindent Observation \emph{(i)} establishes the basis of a proof, by induction on $j$, of the following:

\begin{lemma}\label{lem:1-nest-R}
$R^{jV_{\bY}/(8k)}_{Y_0}(X_0, X_0)$ contains the set of all points that can be reached by robot $\bY$ at location $Y_0$, making $j \le k$ successive moves while another distant neighbour, $\bX$, remains stationary at location $X_0$, provided that each move of $\bY$ is confined to the current $1/k$-scaled safe region with respect to the location of $\bX$.
\end{lemma}
\begin{proof}
For the inductive step, we note that
if $|X_0Y_0| \ge 2 r_1$  and $r_1 \ge r_2$, then
$R^{r_1 + r_2}_{Y_0}(X_0, X_0)$ contains the
union of 
$S^{r_2}_{Y_*}(X_0)$   over all   $Y_* \in R^{r_1}_{Y_0}(X_0, X_0)$. 
This follows from a simple geometric observation captured in Figure~\ref{newfig:LemmaSafe}. Suppose that point $P$ lies in 
$R^{r_1}_{B}(A, A)$. Then the safe region $S^{r_2}_P(A)$ coincides with the safe region $S^{r_2}_{\widehat{P}}(A_{\infty})$, where $\widehat{P}$ denotes the point on $S^{r_2}_P(A)$ that is extreme in the direction 
$\overrightarrow{AB}$, and $A_{\infty}$ denotes the point at infinity in the direction $\overrightarrow{BA}$ 
(see Figure~\ref{newfig:LemmaSafe}.left). 
Assuming $|AB| \ge 2 r_1$, $\angle(A,P,B') < \angle(A,P,\widehat{P}) < \angle(A,P,B)$. 
It follows that $\widehat{P} \in R^{r_1}_{A}(B, B)$.
But the union of $S^{r_2}_{\widehat{P}}(A_{\infty})$, over all points $\widehat{P} \in R^{r_1}_{A}(B, B)$ is exactly 
$R^{r_1 + r_2}_{B}(A, A)$ (see Figure~\ref{newfig:LemmaSafe}.right).
\end{proof}

\begin{figure}[htbp]
\centering
\includegraphics[page=1]{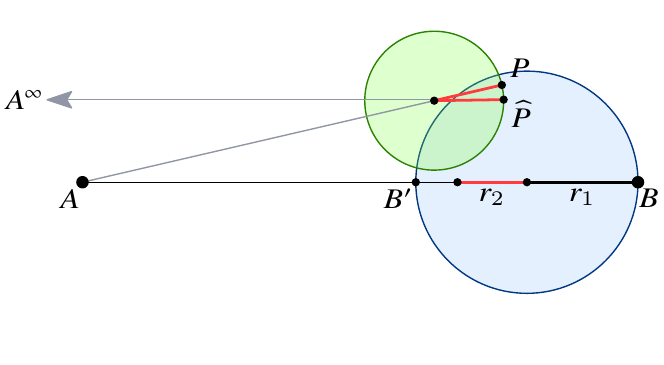}
\hfill
\includegraphics[page=2]{newfig/lemma-safe.pdf}
\caption{Circle expansion lemma. }
\label{newfig:LemmaSafe}
\end{figure}

\noindent Observation \emph{(ii)} establishes the basis of a proof, by induction on $j$, of the following more general lemma:

\begin{lemma}[Base region extension]\label{lem:base-region-extension}
$R^{jV_{\bY}/(8k)}_{Y_0}(X_0, X_1)$ contains the set of all points that can be reached by robot $\bY$ at location $Y_0$, making $j \le k$ successive moves while another distant neighbour, $\bX$, is in the process of moving from location $X_0$ to location $X_1$, provided that each move of $\bY$ is confined to the current $1/k$-scaled safe region with respect to the current location of $\bX$.
\end{lemma}
\begin{proof}
For the induction step, we prove that
for any $\alpha \in (0,1]$, 
$R^{r + \alpha V_{\bY}/8}_{Y_0}(X_0, X_1)$ contains the
union of 
$S^{\alpha V_{\bf Y}/8}_{Y_*}(X_*)$   over all $X_* \in \overline{X_0 X_1}$  and  $Y_* \in R^{r}_{Y_0}(X_0, X_1)$.
The argument here is treated in five cases 
(see Figures~\ref{fig:BaseRegionExtensionLemmaAD}, \ref{fig:BaseRegionExtensionLemmaBE}, and \ref{fig:BaseRegionExtensionLemmaC}, where 
$R^{r}_{Y_0}(X_0, X_1)$ is depicted in blue, 
$R^{r + \alpha V_{\bY}/8}_{Y_0}(X_0, X_1)$ is depicted in red, and the small red circle depicts an extremal safe region associated with a point in $R^{r}_{Y_0}(X_0, X_1)$. 
It is essentially the same as the argument used in 
Lemma~\ref{lem:1-nest-R} above; in all cases, a safe region $S^{\alpha V_{\bf Y}/8}_{Y_*}(X_*)$, where
$X_* \in \overline{X_0 X_1}$  and  $Y_* \in R^{r}_{Y_0}(X_0, X_1)$, coincides with a safe region $S^{\alpha V_{\bf Y}/8}_{\widehat{Y}}(\widehat{X})$, in the extremal form illustrated, where $\widehat{Y} \in R^{r}_{Y_0}(X_0, X_1)$. Since the safe regions in the extreme form illustrated all lie inside  $R^{r + \alpha V_{\bY}/8}_{Y_0}(X_0, X_1)$, the result follows.
\end{proof}

\begin{figure}[htbp]
\begin{minipage}{\textwidth}
\centering
\includegraphics[page=1]{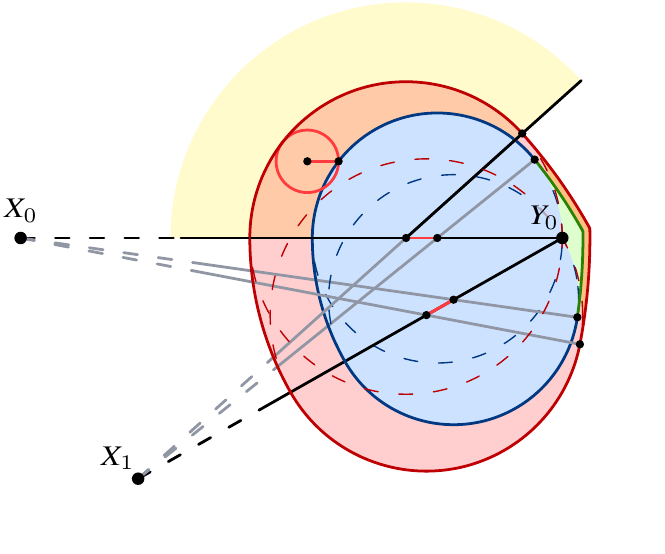}
\hfill
\includegraphics[page=2]{newfig/base-region-extension.pdf}
\caption{{\sc Base region extension} Lemma (cases A and D).}
\label{fig:BaseRegionExtensionLemmaAD}
\end{minipage}
\begin{minipage}{\textwidth}
\centering
\includegraphics[page=3]{newfig/base-region-extension.pdf}
\hfill
\includegraphics[page=4]{newfig/base-region-extension.pdf}
\caption{{\sc Base region extension} Lemma (cases B and E).}
\label{fig:BaseRegionExtensionLemmaBE}
\end{minipage}
\begin{minipage}{\textwidth}
\centering
\includegraphics[page=5]{newfig/base-region-extension.pdf}
\caption{{\sc Base region extension} Lemma (case C).}
\label{fig:BaseRegionExtensionLemmaC}
\end{minipage}
\end{figure}



\section{Visibility preservation by our algorithm, under bounded asynchrony}
\label{sec:VisibilityPreservation}


In this section we begin the analysis of our algorithm by demonstrating the fact that it guarantees that the connectivity of the initial robot configuration is preserved in both the $k$-\nesta and $k$-\async scheduling models. 
This is established by showing that all pairs of robots that are mutually visible in the initial configuration remain so. In addition, we also demonstrate a limited form of preservation of acquired visibility (mutual visibility of robot pairs that arises in subsequent configurations), critical to our congregation argument.

\subsection{Visibility preservation in the \emph{k}-\nesta model 
}\label{subsec:kNestaVisibility}


Though superficially similar, there is an important 
difference between the \ssync and $1$-\nesta models: in the latter, if an active interval of robot $\bY$ is nested within an active interval of robot $\bX$,  it is possible that the \look phase of $\bY$ takes place after the \move phase of $\bX$ has begun. 
In this situation, even if we assume that $\bY$ was seen by $\bX$ when it last looked, 
$\bX$ could be viewed by $\bY$ anywhere on its chosen trajectory.

For the $k$-\nesta model, with $k > 1$, demonstrating that mutual visibility of a pair $\bX,\bY$ of robots is preserved, is 
further complicated by the fact that repeated activations of $\bY$, nested within one activation of $\bX$, may view $\bX$ at many different positions on its determined trajectory, and even if $\bX$ is stationary $\bY$ may view it from up to $k$ different locations. 
In fact, (i) as illustrated by the bulge region, $k$ nested activations of $\bY$, each moving within its $1/k$-scaled safe region with respect to $\bX$, while preserving mutual visibility, can take $\bY$ outside of its basic (unscaled) safe region, and (ii) $k$ nested activations of $\bY$, each moving within even slightly larger safe regions could lead to a break in visibility.




In asynchronous settings, it is certainly possible that a pair of robots under the control of our algorithm could become temporarily mutually visible, and then such a visibility may become subsequently lost. However, a critical component of our congregation argument relies on the fact that, if the separation of a robot pair becomes sufficiently small, visibility will be preserved indefinitely.
We say that robot $\bX$ is  \emph{strongly visible} from robot $\bY$ at time $t$ if $|X(t) Y(t)| \le V/2$. 

With this definition in hand we can state and prove what we need in terms of visibility preservation for the $k$-\nesta scheduling model:

\begin{theorem} \label{thm:VisibilityPreservationNesta}
Provided robots $\bX$ and $\bY$ both
continue to confine their movement to the $1/k$-scaled safe regions for motion with respect to the other, then if
(i)  $\bX$ is  visible from $\bY$ at time $t=0$ (i.e. $|X(0) Y(0)| \le V$), or 
(ii) $\bX$ becomes  strongly visible from $\bY$ at some subsequent time $t>0$, 
then their mutual visibility will be maintained thereafter,
under arbitrary $k$-\nesta scheduling.
\end{theorem}

\noindent The proof of the two cases in this theorem is developed in the following two subsections.

\subsubsection{Preservation of initial visibilities}

We begin with the preservation of initial visibilities. 
Specifically, in the following we prove that
any pair of robots $\bX$ and $\bY$, satisfying $|X(0) Y(0)| \le V$, also satisfies $|X(t) Y(t)| \le V$, for all $t \ge 0$.


\begin{proof}

Recall that, in the $k$-\nesta model, the schedule of activations of any pair of robots $\bX$ and $\bY$ decomposes into a sequence activation events which consist of a single activation interval of one of the two robots within which are nested up to $k$ activation intervals of the other. Hence, it will suffice to prove that (i) if $\bX$ and $\bY$ are mutually visible at the start of any such 
activation event then they remain mutually visible throughout the event, and (ii) if one of the robots becomes strongly visible from the other during the course of the activation event then they will remain mutually visible thereafter. 
In fact, since the diameter of all $1/k$-scaled  safe regions is at most $V/(4k)$, the total length of moves of $\bX$ and $\bY$ in any activation event is at most $V/2$, and thus one robot can become strongly visible from the other during the course of the activation event only if they are mutually visible at both the start and end. Hence, it suffices to prove property (i) alone, knowing that at no intermediate point one of the robots becomes strongly visible from the other.

Suppose that an activation event begins with the pair of robots $\bX$ and $\bY$ at locations $X_0$ and $Y_0$ respectively, with separation at most $V$. Suppose further that this  activation event consists of a single active interval of robot $\bX$ and some $j \le k$ active intervals of robot $\bY$  all of which are nested within the active interval of $\bX$ (in particular, at both the start and end of the active interval of $\bX$, both $\bX$ and $\bY$ are immotile).
Finally, suppose that (i) the active interval of robot $\bX$ determines a target destination $X_1$ that is confined to the $1/k$-scaled  safe region $S_{X_0}^{V_{\bX}/(8k)}(Y_0)$, 
and (ii) the $i$-th active intervals of robot $\bY$ determines a target destination $Y_i$
that is confined to the $1/k$-scaled safe region $S_{Y_{i-1}}^{V_{\bY}/(8k)}(X_*)$ for some $X_* \in \overline{X_0 X_1}$ 
(recall that relocation to every point within a safe region may not be realizable in practice, since moves are ultimately dependent on the locations of one or more neighbouring robots).

Now suppose, leading to a contradiction, that at the end of the active interval of robot $\bX$, robots $\bX$ and $\bY$ are no longer mutually visible. In this case, we can assume that we are dealing with a counterexample with minimum $j$.
By definition, $\bY$, at location $Y_0$, is visible from $\bX$ at the start of its  interval. Furthermore, at the start of its first nested active interval, $\bY$, still at location $Y_0$, must view $\bX$ at some location $X_*$ on its trajectory towards $X_1$.  
Since $S_{X_0}^{V_{\bX}/(8k)}(Y_0)$ is entirely contained in the disk of radius $V$ centered at $Y_0$, the point $X_*$ is visible from robot $\bY$ at this point in time.
To establish the preservation of mutual visibility, we need to show that 
$\bX$ remains visible to $\bY$ at the time $\bY$ ends each of its successive nested active intervals. Of course, by the minimality of our supposed counterexample, we can assume that $\bX$ remains visible to $\bY$ at the time $\bY$ ends each of its first $j-1$ nested active intervals.

But we have already established, in Lemma~\ref{lem:1-nest-R}, that the set of points reachable by $\bY$ starting at $Y_0$ making $j$ successive moves each of which is confined to an $1/k$-scaled safe region with respect to some (possibly changing) location of $\bY$ within the segment $\overline{X_0X_1}$, is contained in $R^{jV/(8k)}_{Y_0}(X_0, X_1)$.  Given this, it is straightforward to confirm that the distance from $X_1$ to $Y_j$ is at most $|X_1 Y_0^+|$, which is at most $|X_0 Y_0|$. Hence $\bX$ remains visible to $\bY$ at the end of its $j$-th active interval.
\end{proof}

\subsubsection{Preservation of acquired visibilities}

To complete the proof of Theorem~\ref{thm:VisibilityPreservationNesta}
it remains to argue that if robot
$\bX$ becomes  strongly visible from $\bY$ at some time $t>0$, 
then their mutual visibility will be maintained thereafter,
under arbitrary $k$-\nesta scheduling.

\begin{proof}

We have already established that if a pair of robots $\bX$ and $\bY$ are mutually visible at the start of an active interval for $\bX$ that contains up to $k$ nested active intervals for $\bY$, then they remain mutually visible up to the end of the active interval for $\bX$. So if $\bX$ and $\bY$ have their visibility broken at some point it must be that they had a separation greater than $V$ at the start of the active interval for $\bX$. But since both $\bX$ and $\bY$ move a distance at most $V/4$ over the duration of this active interval, it follows that their separation never falls below $V/2$. Hence, if the separation of $\bX$ and $\bY$ ever falls below $V/2$ in this interval it must be that $\bX$ and $\bY$ are mutually visible at the end, and thereafter. 
\end{proof}


\subsection{Visibility preservation in the \emph{k}-\async model
}

Turning now to the more inclusive $k$-\async scheduling model, we establish the following generalization of Theorem~\ref{thm:VisibilityPreservationNesta}, exploiting the structure common to the $k$-\nesta and $k$-\async models.

\begin{theorem} \label{thm:VisibilityPreservationAsync}
Provided robots $\bX$ and $\bY$ both
continue to confine their movement to the $1/k$-scaled safe regions for motion with respect to the other, then if
(i)  $\bX$ is  visible from $\bY$ at time $t=0$ (i.e. $|X(0) Y(0)| \le V$), or 
(ii) $\bX$ becomes  strongly visible from $\bY$ at some subsequent time $t>0$, 
then their mutual visibility will be maintained thereafter,
under arbitrary $k$-\async scheduling.
\end{theorem}

As in our proof of Theorem~\ref{thm:VisibilityPreservationNesta}, we treat the preservation of initial visibilities and acquired visibilities separately.
In each of these cases, we begin by treating the case $k=1$, which captures many of the most important features of the general argument.
Treatment of the case where $k > 1$, which builds on this simpler case, and exploits many of the same ideas that arose in Theorem~\ref{thm:VisibilityPreservationNesta}, will follow. 
%
\subsubsection{Preservation of initial visibilities}
Suppose that robots $\bX$ and $\bY$ have separation at most $V$ in their initial configuration (when both are immotile).
We need to show that, from this point onward,  provided  $\bY$ moves to a point within the safe region with respect to the observed position of $\bX$, the visibility of $\bX$ will not be lost, no matter how $\bX$ chose to move, \emph{based on its observed position of $\bY$}. Note, the argument here seems to be intrinsically more difficult than that used in the $1$-\nesta model.
In particular, it does not suffice, as it did in the $k$-\nesta model, to focus attention on just one active interval for one of the two robots whose mutual visibility is being considered. In the $1$-\async model, the validity of each step depends on the full history of transitions since the last time both robots were simultaneously immotile.\\

\noindent{\bf 1-\async Case.}\\
\indent We say that a pair of robots $\bX$ and $\bY$ is {\em engaged} at time $t$ if at least one of them is active. Suppose that $\bX$ and $\bY$ become engaged at time $t _0$,
with $|X(t_0) Y(t_0)| \le V$. We proceed to show that while $\bX$ and $\bY$ remain engaged their separation remains at most $V$.
Suppose, leading ultimately to a contradiction, that
\begin{enumerate}
\item  
at time $t_0$, $\bX$ is starting an active interval and
$\bY$ is immotile,
i.e. either idle or starting an active interval, 
\item  
each subsequent move of $\bX$ and $\bY$, the timing and extent of which is determined by an adversarial scheduler, has a target destination that respects the 
safe regions of $\bX$ and $\bY$ alone (and hence may 
form a strict superset of the realizable destinations when other robots are taken into consideration); and
\item  
at some subsequent time 
$t_*$, during the $i$-th active interval of $\bX$ and after the start of the $h$-th active interval, and before the start of the $(h+1)$-th active interval, of $\bY$, 
and prior to which $\bX$ and $\bY$ remain engaged,
$\bX$ and $\bY$ reach positions $X(t_*)$ and $Y(t_*)$ whose separation is greater than $V$.
\end{enumerate}
If the schedules of $\bX$ and $\bY$
are truncated at time $t_*$, we refer to the interval  $[t_0, t_*]$ as a {\em doomed engagement} of $\bX$ and $\bY$. There is, of course, no loss of generality in assuming that both $\bX$ and $\bY$ are active just prior to time $t_*$.
We will assume that the doomed engagement under consideration minimizes $i+h$.
For $a \ge 1$, we define $X_a$ (respectively, $Y_a$) to be the position of robot $\bX$ (respectively $\bY$) at the end of its $a$-th active interval starting at or after time $t_0$. 
For convenience in subsequent definitions, we denote $X(t_0)$ (respectively $Y(t_0)$) by $X_0$ (respectively $Y_0$), and imagine that $\bY$ has a $0$-th active interval intersecting $t_0$ during which it is stationary (and accordingly define $Y_{-1}= Y_0$).

We first note that, by minimality, 
\begin{enumerate}
\item
During the doomed engagement,
the $\bX$ and $\bY$ active intervals interleave, i.e.  (i) for all $1 \le j \le i$, the $j$-th activation interval of $\bY$ starts during the $j$-th active interval of $\bX$,
and (ii) for all $1 < j \le i$, the $j$-th
active interval of $\bX$ starts during the $(j-1)$-st active interval of $\bY$ (see Figure~\ref{newfig:InterleavedActivations}). Minimality dictates that the active intervals overlap, and by our $1$-\nesta analysis they cannot nest.

\begin{figure}[htbp]
\centering
\includegraphics{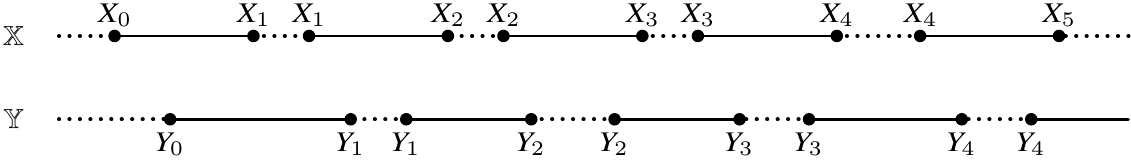}
\caption{Without loss of generality, active intervals for $\bX$ and $\bY$ interleave, up to disconnection. }
\label{newfig:InterleavedActivations}
\end{figure}

\item
The doomed engagement ends 
at the end of the $i$-th activation of both $\bX$ and $\bY$);
when $\bX$ is at location $X_{i}$ and $\bY$ is at location $Y_{i}$ (the case where 
the doomed engagement ends
at the end of the $i+1$-st activation of  $\bX$ and the $i$-th activation of $\bY$ is handled identically, with only a trivial change in indexing).

\item
For all $0 \le j \le i$, the separation between $X_j$ and $Y_{j-1}$, as well as that between $Y_{j-1}$ and $X_{j-1}$, is at most $V$.

\item 
For all $1 \le j < i$, the point $Y_{j+1}$ lies in the basic safe region 
$S^{V_{\bY}/8}_{Y_{j}}(X_*)$ 
associated with $\bY$ at location $Y_{j}$, viewing $\bX$ at some location $X_*$ between $X_{j}$ and $X_{j+1}$. 
Similarly, the point $X_{j+1}$ lies in the basic safe region $S^{V_{\bX}/8}_{X_{j}}(Y_*)$
associated with $\bX$ at location $X_{j}$, viewing $\bY$ at some location $Y_*$ between $Y_{j-1}$ and $Y_{j}$. 
We refer to these as the \emph{safety constraints} for points $Y_{j}$ and $X_{j}$, respectively.
\end{enumerate}
We will analyze the chain of edges 
$$\overline{Y_{i} X_{i}},\; \overline{X_{i} Y_{i-1}},\; \overline{Y_{i-1}X_{i-1}},\; \overline{X_{i-1}Y_{i-2}},\ldots, \overline{X_2Y_1},\; \overline{Y_1 X_1},\; \overline{X_1 Y_{0}},\; \overline{Y_0 X_0},\;
\overline{X_0 Y_{-1}}
$$ 
associated with a minimal doomed engagement for $\bX$ and $\bY$,
essentially a walk from their terminal configuration back to their initial configuration, ultimately showing that no such chain exists that satisfies all of the specified constraints, thereby demonstrating our contradiction. 

For $t \ge 0$, let $e_t$ denote the $t$-th edge on this chain; i.e., 
$e_t = \overline{Y_{i- t/2 } X_{i-t/2 }}$, for even $t$, and
$e_t = \overline{X_{i -(t-1)/2}Y_{i-1 -(t-1)/2}}$, for odd $t$,
and
$\theta_t$ --- the (clockwise) angle, in the range $(-\pi, \pi]$ between $e_{t+1}$ and $e_{t}$. 
Note that $\theta_{2i} = 0$ but, by our minimality assumption, 
$\theta_t \neq 0$, for $0 < t < 2i$.
The minimality assumption implies other structural constraints as well:


\begin{enumerate}

\item 
Observe that, in order for the distance between $Y_{i}$  and $X_{i}$  to be at least $V$, 
it must be the case that 
$\theta_0 \neq 0$ (otherwise, by minimality of our example, $Y_i$ must lie on the line through $X_i$ and $Y_{i-1}$ but outside the segment $\overline{X_i Y_{i-1}}$, and so outside of the safe region 
$S^{ V_{\bY}/8}_{Y_{i-1}}(X_*)$
the safe region for $\bY$ at location $Y_{i-1}$ viewing $\bX$ at any location $X_*$ on the segment $\overline{X_{i-1} X_i}$). Thus, in particular,  $i$ must be greater than $1$; 

\item 
Next, note that, for all $0< j< i$, $Y_{j+1}$ is \emph{not} within the safe region for $\bY$ at location $Y_j$ viewing $\bX$ at location $X_{j+1}$ (and $X_{j+1}$ is \emph{not} within the safe region for $\bX$ at location $X_j$ viewing $\bY$ at location $Y_{j}$) (otherwise, our minimal separating sequence could be shortened, by starting instead at $X_{j+1}$ and $Y_j$, or $X_j$ and $Y_j$, a contradiction).

\begin{figure}[htbp]
	\centering
	\includegraphics{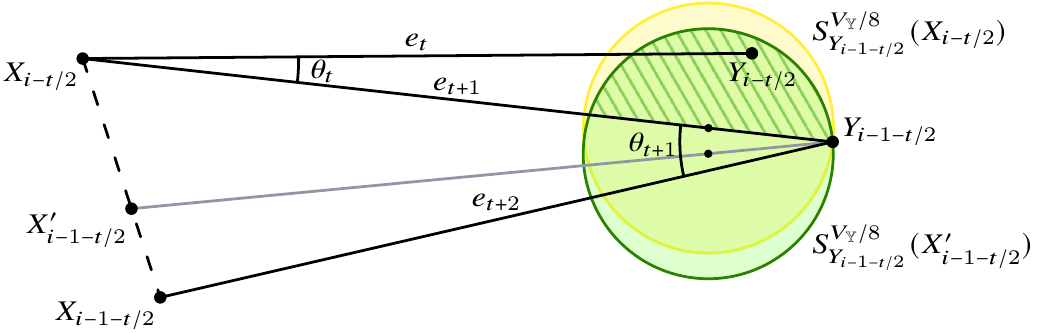}
	\caption{
		Illustration for the case when $\theta_t>0$ and $\theta_{t+1}<0$.
		Region $S^{V_\bY/8}_{Y_{i-1-t/2}}(X'_{i-1-t/2})$ (in green) is the safe region of $\bY$ at location $Y_{i-1-t/2}$ when it sees $\bX$ at some point $X'_{i-1-t/2}\in\overline{X_{i-1-t/2}X_{i-t/2}}$.
		Region $S^{V_\bY/8}_{Y_{i-1-t/2}}(X_{i-t/2})$ (in yellow) is the safe region of $\bY$ at location $Y_{i-1-t/2}$ when it sees $\bX$ at $X_{i-t/2}$.
		As the striped region is fully contained in $S^{V_\bY/8}_{Y_{i-1-t/2}}(X_{i-t/2})$, any location $Y_{i-t/2}$ that is in $S^{V_\bY/8}_{Y_{i-1-t/2}}(X'_{i-1-t/2})$ is also in $S^{V_\bY/8}_{Y_{i-1-t/2}}(X_{i-t/2})$.}
	\label{newfig:AngleConsistencyNEW}
\end{figure}

\item
It follows that 
 we can assume that 
 $\theta_t > 0$, for $0 < t < 2i$.
Certainly this holds without loss of generality for $\theta_0$ (by replacing clockwise angles with counterclockwise, if necessary). 
 Moreover, if it holds for $\theta_t$ then it must hold for $\theta_{t+1}$.
Suppose, to the contrary, that $\theta_t >0$ and 
$\theta_{t+1} < 0$, when $t$ is even (see Figure~\ref{newfig:AngleConsistencyNEW}); the argument when $t$ is odd is essentially the same. 
Then, since $Y_{i -t/2}$ belongs to the safe region for $\bY$ at location $Y_{i-1- t/2 }$ when it sees $\bX$ at some location on the segment $\overline{X_{i-1-t/2 } X_{i-t/2 }}$, $Y_{i -t/2}$ must also belong to the safe region for $\bY$ at location $Y_{i-1- t/2 }$ when it sees $\bX$ at 
location $X_{i-t/2 }$. 
(This follows because, restricted to points $Z$ for which 
$\angle(Y_{i-1- t/2}, X_{i-t/2 }, Z ) >0$,
the safe region for $\bY$ at location $Y_{i-1- t/2 }$ viewing $\bX$ at any location on $\overline{X_{i-1-t/2 } X_{i-t/2 }}$ is contained in the safe region for $\bY$ at location $Y_{i-1- t/2 }$ viewing $\bX$ at location $X_{i-t/2}$;
This, contradicts our minimality assumption, as noted previously.

\end{enumerate}

It is not hard to see that, in order to satisfy local safety constraints, there is a kind of tradeoff between 
$|e_t|$ 
and 
$\theta_t$: specifically, if $|e_t|$ is increased then the safety constraint at the common endpoint of $e_t$ and $e_{t+1}$ can be satisfied with a smaller angle $\theta_t$.
We make this tradeoff explicit 
in the following:

\begin{lemma}\label{lem:tradeoff}
$|e_t| > V \cos \theta_t $ and 
$\cos \theta_t \ge \sqrt{(2+\sqrt{3})/4} \approx 0.9659$, for all $t \ge 0$.
\end{lemma}

Note that since the chain ends with $\theta_{2i} = 0$ (because $\bY$ is stationary in its $0$-th active interval), it follows that the edge $e_{2i} = \overline{X_0 Y_{0}}$ has length greater than $V$, a contradiction that completes the proof of Theorem~\ref{thm:VisibilityPreservationAsync} in the case $k=1$.

\begin{proof}

The proof of Lemma~\ref{lem:tradeoff} proceeds by induction on $t$.
 

    Since $|e_0|$, the length of edge $\overline{X_{i} Y_{i}}$, is at least $V$, it follows immediately that $|e_0| \ge V \cos \theta_t$. 
    To complete the basis of our induction, observe that 
    in order to satisfy the safety constraint at $Y_{i-1}$ it must be the case that $Y_{i-1}$ is no further than $V/4$ from $Y_i$, and hence no further than $V/4$ from $e_0$. 
    It follows that $|e_0| \sin \theta_0 \le V/4$, and thus
    $\sin \theta_0 \le 1/4$, which holds exactly when
    $\cos \theta_0 \ge \sqrt{15/16} > \sqrt{(2+\sqrt{3})/4}.$

    Suppose now that 
    $|e_t| > V \cos \theta_t $ and 
$\cos \theta_t \ge \sqrt{(2+\sqrt{3})/4}$.
    We will show that 
    $|e_{t+1}| > V \cos \theta_{t+1} $ and 
$\cos \theta_{t+1} \ge \sqrt{(2+\sqrt{3})/4}$.
For this induction step, consider four successive points on the chain 
(renamed $S$, $R$, $Q$, and $P$, for simplicity and to reflect the commonality of the argument for even and odd $t$),
their associated edges $\overline{SR}$, $\overline{RQ}$ and $\overline{QP}$, and their intervening angles (renamed $\beta$ and 
$\varphi$) 
(see Figure~\ref{newfig:anlge-divergence}).
We show that assuming $|SR| > V \cos \beta$ and
$\cos \beta \ge \sqrt{(2+\sqrt{3})/4}$, then
$|RQ| > V \cos \varphi$ and $\cos \varphi \ge \sqrt{(2+\sqrt{3})/4}$.

\begin{figure}[htbp]
	\centering
	\includegraphics[page=3]{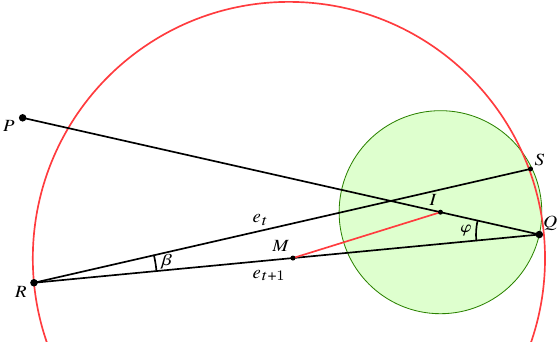}
	\caption{
The point $M$ is at distance $V/2$ from $R$, and is the center of the circle $c_{V/2}$ (in red). The safe region $S^{V/8}_Q(P)$ is shown in green. The two circles must intersect for the conditions on point $S$ to hold.}
	\label{newfig:anlge-divergence}
\end{figure}

First observe that motion from $P$ to $R$ and from $Q$ to $S$ are the result of single move steps of the corresponding robots. 
As we shall see in Section~\ref{sec:IncrementalConvergence}, where we specify a robot's destination as a function of the locations of its visible neighbours, this means that $|PQ|$ and $|QS|$ are both at most $V/8$.
It follows then that both $|PQ|\ge |RS| - V/4 > V/2$ and 
$|RQ|\ge |RS| - V/8 > V/2$; since $|RS|\ge\sqrt{(2+\sqrt{3})/4}$. In particular, the robot at location $Q$ sees the robot at position $P$ as a distant neighbour when planning its motion to location $S$.

Consider a circle $c_{V/2}$ of radius $V/2$ centered at a point $M$ lying on the segment $\overline{RQ}$ at distance $V/2$ from $R$
(see Figure~\ref{newfig:anlge-divergence}).
Point $S$ must be outside of $c_{V/2}$, otherwise $|SR| \le V\cos \beta$ (due to the angle $\angle RSQ' \ge \pi/2$ where $Q'$ is an intersection of the extension of $\overline{RQ}$ with $c_{V/2}$); hence $|MS| > V/2$.
Let $I$ be the center of $S^{V/8}_Q(P)$, i.e., $I$ lies on $\overline{QP}$ at distance $V/8$ from $Q$. 
Since point $S$ lies in $S^{V/8}_Q(P)$, $|IS| \le V/8$ and hence $|MI| \ge |MS|-|IS| > 3V/8$. 

Consider the triangle $\triangle QIM$.
From the cosine theorem it follows that 
\[
|MI|^2= |MQ|^2 + |IQ|^2 - 2|MQ| |IQ| \cos\varphi\,.
\] 
Then, after dividing by $V^2$ and replacing $|MQ|/V$ by $x$,  the equation above becomes
\[
x^2 - 2x \frac{|IQ|}{V} \cos\varphi + \left(\frac{|IQ|}{V}\right)^2 - \left(\frac{|MI|}{V}\right)^2 = 0\,.
\]
Since $x>0$, it follows that $x = \frac{|IQ|}{V} \cos\varphi + \sqrt{ \left(\frac{|IQ|}{V}\right)^2 \cos^2 \varphi - \left(\frac{|IQ|}{V}\right)^2 + \left(\frac{|MI|}{V}\right)^2}$. 
But since $|IQ|=V/8$, $|MI| > 3V/8$, and $\cos\varphi \le 1$, we get
\[
x >  \frac{\cos\varphi}{8} + \sqrt{ \frac{\cos^2 \varphi}{64} - \frac{1}{64} + \frac{9}{64}}
   \ge  \frac{\cos\varphi}{8} + \sqrt{ \frac{\cos^2 \varphi}{64} - \frac{\cos^2 \varphi}{8}}
    =    \frac{\cos\varphi}{8} + \frac{3\cos\varphi}{8} = \frac{\cos\varphi}{2}\,.
\]
%
Thus,  $\frac{|RQ|}{V} = \frac{1}{2} + \frac{|MQ|}{V} > \frac{1}{2} + \frac{\cos\varphi}{2} \ge \cos\varphi$.

Now, as $|PR|\le V/4$ and the length of the perpendicular dropped from $R$ onto the line supporting $\overline{PQ}$ is at most $|PR|$, we obtain that
\[
\sin\varphi \le \frac{V}{4|QR|} \le \frac{1}{4\cos\varphi}\,,
\]
or $\sin^2\varphi \cos^2\varphi \le 1/16$ from which it follows that $\cos^2\varphi  \ge \frac{2+\sqrt{3}}{4}$ or $\cos\varphi\ge\sqrt{\frac{2+\sqrt{3}}{4}}$.
\end{proof}

\noindent {\bf $k$-\async Case, when  ${k > 1}$.}\\
\indent As we saw in the $k$-\nesta model, we take advantage of the fact that up to $k$ motions within safe regions that have been scaled by a factor of $1/k$, are largely equivalent to single unscaled motions.
Accordingly, the proof of Theorem~\ref{thm:VisibilityPreservationAsync} when $k > 1$ is essentially a reduction to the case $k=1$: we argue that if there is a separating sequence when $k >1$, there must also be one when $k=1$ (which we have just shown is impossible). The proof proceeds as follows:

\begin{enumerate}
    \item 
    As in the case where $k=1$, we first suppose that a separating sequence for two robots $\bX$ and $\bY$ exists, and define a minimal separating sequence to be one that minimizes the total number of activations of $\bX$ and $\bY$.  As before, it follows from minimality that the activations of $\bX$ and $\bY$ interleave, with up to $k$ activations of one, an \emph{activation cluster}, between activations of the other. 
    \item 
    We denote by $X_{j,1}, \ldots, X_{j,s_j}$, where $s_j \le k$, the positions of robot $\bX$ at the start of the first activation in its $j$-th activation cluster, which occurs as robot $\bY$ is executing the last activation cycle of its $(j-1)$-st activation cluster.
    (Similarly, we denote by $Y_{j,1}, \ldots, Y_{j,t_j}$, where $t_j \le k$, the positions of robot $\bY$ at the start of the first activation in its $j$-th activation cluster, which occurs as robot $\bX$ is executing the last activation cycle of its $j$-th activation cluster
    (see Figure~\ref{newfig:kInterleavedActivations}).
    
    \begin{figure}[htbp]
\centering
\includegraphics[page=2]{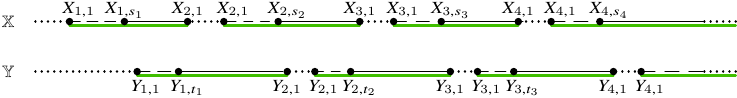}
\caption{Without loss of generality, active intervals for $\bX$ and $\bY$ interleave, up to disconnection, with at most $k$ activations of one between activations of the other. Activation clusters are shown in green.}
\label{newfig:kInterleavedActivations}
\end{figure}

    \item We refer to the locations at the start of activation clusters as \emph{activation checkpoints}, and observe that there is no loss of generality in assuming that the first point at which a separation occurs is at the separation checkpoint pair $(X_{i+1, 1}, Y_{i+1, 1})$. 
    
    \item 
    The same \emph{safety constraints} that applied in the case where $k=1$ are expressed more generally as follows.
    For all $j \le i$,
    \begin{enumerate}[(i)]
        \item for $1 \le t < t_j$, the point $Y_{j, t+1}$ lies in the $1/k$-fraction safe region associated with $\bY$ at location $Y_{j,t}$, viewing $\bX$ at some location between $X_{j, s_j}$ and $X_{j+1, 1}$, and
        \item the point $Y_{j+1, 1}$ lies in the $1/k$-fraction safe region associated with $\bY$ at location $Y_{j,t_j}$, viewing $\bX$ at some location between $X_{j, s_j}$ and $X_{j+1, 1}$.
    \end{enumerate}
    Similarly, 
    \begin{enumerate}[resume*]
        \item for $1 \le s < s_j$, the point $X_{j, s+1}$ lies in the $1/k$-fraction safe region associated with $\bX$ at location $X_{j,s}$, viewing $\bY$ at some location between $Y_{j-1, s_{j-1}}$ and $Y_{j, 1}$, and
        \item the point $X_{j+1, 1}$ lies in the $1/k$-fraction safe region associated with $\bX$ at location $X_{j,s_j}$, viewing $\bY$ at some location between $Y_{j-1, s_{j-1}}$ and $Y_{j, 1}$.
    \end{enumerate}
    \item 
    If the sequence of points $X_{j, 2}, \ldots, X_{j, s_j}$ all lie on the segment $X_{j, 1} X_{j+1, 1}$ we say that the $j$-th activation cluster of $\bX$ is \emph{linear}. 
    Similarly, if the sequence of points $Y_{j, 2}, \ldots, Y_{j, t_j}$ all lie on the segment $Y_{j, 1} Y_{j+1, 1}$ we say that the $j$-th activation cluster of $\bY$ is \emph{linear}. 
    
\tab We can assume that in our hypothesized minimum separating sequence all activation clusters are linear. This follows by induction on $j$ (the index of the activation cluster) using the facts that (i) during the initial activation cluster of $\bX$, robot $\bY$ is fixed at location $Y_{1,1}$, and (ii) assuming the other robot is following a linear trajectory between activation checkpoints, there is no loss of generality in assuming that the next activation cluster is linear.  The latter exploits our characterization, from our analysis in the $k$-\nesta setting, of the locus of points reachable by up to $k$ successive $1/k$-fraction moves of one robot seeing the other at successive points on some linear trajectory, in particular the fact that (by convexity of the regions involved) all reachable points are reachable by a straight path.
    
    \item 
    Given that all activation clusters in our separating sequence are linear, it becomes evident that the sequence of transitions between successive activation checkpoints 
    corresponds closely to a valid separating sequence in the $1$-\async model.
    We then can extend the proof of Lemma~\ref{lem:tradeoff} to the case when $k>1$.
    In it, our characterization of the full safe region (recall Lemma~\ref{lem:base-region-extension}) is slightly more complicated (and slightly larger) than assumed for the case $k=1$.
    In fact, our analysis in Lemma~\ref{lem:tradeoff} (in the case $k=1$) treats the situation where point $S$ lies in the core of the safe region for the robot at location $Q$.
    
    \tab We show that $|RQ| > V \cos \varphi$ in the second case as well, when $S$ lies in the bulge (see Figure~\ref{newfig:anlge-divergence-k}). 
    Observe that still $|MI|>3V/8$ and that, consequently, our earlier proof continues to hold, despite the fact that it is no longer the case that $|IS| \le V/8$.
    Indeed, let $H$ be the intersection point of a ray emanating from $I$ along $\overrightarrow{RI}$ and circle $S^{V/8}_Q(P)$.
    Note that $|IH| = V/8$ and $|MH|>V/2$ (otherwise point $H$ and, thus, point $S$ do not lie outside of circle $c_{V/2}$, contradicting our assumption that $|SR| > V\cos \beta$).
    Thus $|MI|\ge |MH|-|IH| > 3V/8$, and Lemma~\ref{lem:tradeoff} holds for $k>1$.

\end{enumerate}


\begin{figure}[htbp]
	\centering
	\includegraphics[page=4]{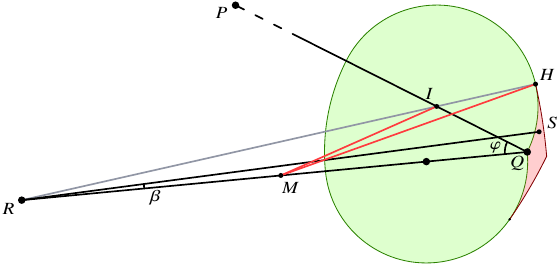}
	\caption{
		Region $R^{V/8}_Q(R,P)$ intersects $c_{V/2}$ only if $S^{V/8}_Q(P)$ intersects $c_{V/2}$.
	}
	\label{newfig:anlge-divergence-k}
\end{figure}

\subsubsection{Preservation of acquired visibilities}

Here it suffices to observe that if at some point $\bY$ is seen by $\bX$ at some distance less than $V/2$, then $\bX$ and $\bY$ must remain mutually visible thereafter. This follows because at the next checkpoint robots $\bX$ and $\bY$ each could not have moved more than distance $V/8$, and hence their separation could not exceed $3V/4$ at that point. Hence, by our inductive assertions established in the preceding subsubsection that demonstrate that any configuration of $\bX$ and $\bY$ leading to a separation of at least $V$ must be preceded by configurations of $\bX$ and $\bY$ with separation at least $0.96 V$, we conclude that $\bX$ and $\bY$
must remain mutually visible.

\section{Incremental congregation in the \emph{k}-\nesta and \emph{k}-\async models
}
\label{sec:IncrementalConvergence}


We have, to this point, not specified the details of the destination function of an activated robot $\,\bZ$, other than that it must be a point in the intersection of the safe regions with respect to all of its distant neighbours (those neighbours whose distance is in the range $(V_{\bZ}/2, V_{\bZ}]$, where $V_{\bZ}$ is the distance to the furthest visible neighbour of $\,\bZ$). We choose the center point of this intersection as the target destination.

In the event that the distant neighbours are not properly contained in any halfspace with $\,\bZ$ on its boundary, the intersection of the corresponding safe regions is just the current location of $\,\bZ$, so the target destination is just this current location (i.e. $\,\bZ$ does not move).
Other events consist of two cases:
\begin{enumerate}[(i)]
\item $\bZ$ has only one distant neighbour. In this case the target destination is just the center point of the sole constraining safe region.

\item $\bZ$ has two or more distant neighbours. In this case the target destination is the center point of the intersection of the safe regions associated with the two distant neighbours that define the largest sector containing all of the distant neighbours of $\,\bZ$ 
(see Figure~\ref{newfig:MotionSpecification}).
Note that the center point is the middle point of the segment connecting the centers of the safe regions corresponding to the two distant neighbours.
\end{enumerate}
\begin{figure}[tbp]
\centering
\includegraphics[page=1]{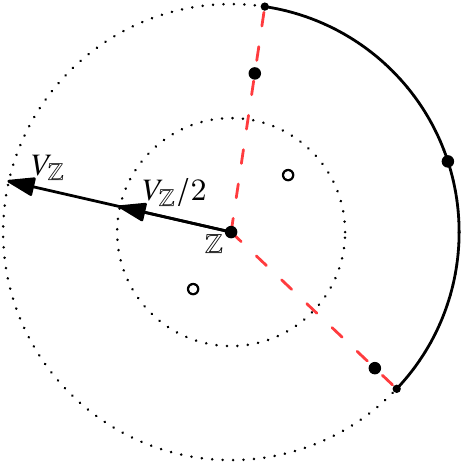}
\hfill
\includegraphics[page=2]{newfig/motion.pdf}
\hfill
\includegraphics[page=3]{newfig/motion.pdf}
\caption{Visualization of the target destination $x$ for $\bZ$ when it admits two distant neighbours.}
\label{newfig:MotionSpecification}
\end{figure}%
Note that, in all events, the distance to the target destination of an activated robot is at most $V/8$.

As we have demonstrated, our algorithm has the property that it preserves all initial visibilities 
as well as all visibilities that at some point fall below the $V/2$ threshold. 
We say that a pair of robots are {\em strong neighbours} at some point in time if one of these conditions holds. 
It follows that the graph of the strong neighbour relation is both connected, because it includes all initially visible pairs, and monotonic. In particular, 
at any moment in time there is a path in this graph joining any pair of robots.
Furthermore there is a point in time after which the graph no longer changes:
thereafter for any pair of robots, (i)~their separation is continuously less than or equal to $V$ (i.e., they are mutually visible), or (ii)~their separation is continuously greater than $V/2$. This, combined with the fact that our strategy is hull-diminishing, suffices to show that the robots always converge to a point, from any initial configuration.

We begin with an informal overview of the argument.
\begin{itemize}
\item
As with all hull-diminishing strategies, progress towards convergence to a point can be measured in terms of the shrinkage of the convex hull of the full set of robot locations. Let $CH_t$ denote the convex hull of the robot locations, including their planned but as yet unrealized trajectories, at time $t$.  Then incremental progress towards convergence is captured by the following observation (made previously in~\cite{K11}): 

\begin{center}
\smallskip
\begin{minipage}{0.8\columnwidth}
$CH_{t^+} \subseteq CH_t$ for all $t^+ > t$.
This follows from the facts that (i) when a motion is planned it is towards a point inside the current convex hull, and (ii) when the motion completes, and the planned trajectory (or some prefix) has been realized, replacing the trajectory by its endpoint can only shrink the convex hull.
\end{minipage}
\smallskip
\end{center}


\item
Since convex hulls form a nested
sequence, they either converge to a single point (which is what we are trying to establish), or not. Suppose that the convex-hull sequence does not converge in this way for some initial configuration of $n$ robots. Then both the convex hull perimeter and the hull radius (the radius of the smallest ball enclosing the convex hull) must form monotonically decreasing sequences that converge to some non-zero values, $\Lambda$ and $\rho$ respectively. 
Thus, for any $\varepsilon > 0$ there is a time $t_{\varepsilon}$ after which the hull perimeter lies in the range $[\Lambda, \Lambda + \varepsilon)$ 
and the hull radius lies in the range $[\rho, \rho + \varepsilon)$.

\item
Let $t$ be any time after which the hull radius lies in the range $[\rho, 3\rho/2)$  and consider the convex hull $H = CH_t$, with perimeter $\Lambda_t$, and its smallest bounding circle $\,\Xi$, with center $O_H$ and radius $r_H$, at that time (see Figure~\ref{fig:LimitingHull}). 
We will argue that  there is a $\lambda > 0$, dependent on $\rho$, such that at some time $t' > t$ the perimeter of $CH_{t'}$ becomes less than $\Lambda_t-\lambda$. This leads to a contradiction, if we choose $\varepsilon < \lambda$ and $t \ge t_{\varepsilon}$.

\end{itemize}

\begin{figure}[tbp]
\centering
\includegraphics[]{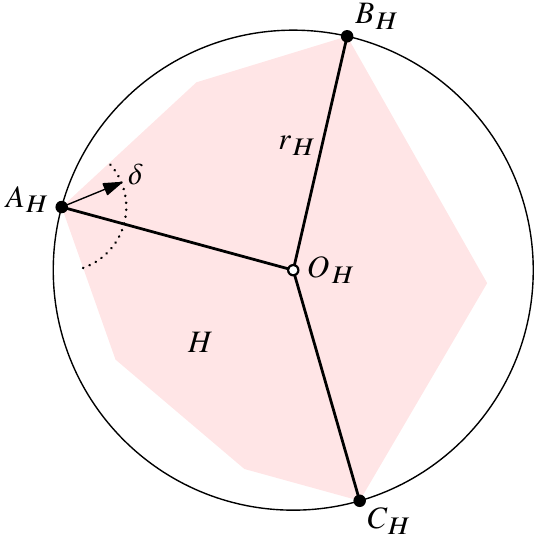}
\caption{Smallest bounding circle $\,\Xi$.}
\label{fig:LimitingHull}
\end{figure}

We know that $\,\Xi$ is determined by three critical points (robot positions) $A_H$, $B_H$ and $C_H$ (two of which might coincide), all of which are at distance exactly $r_H$ from $O_H$, and in pairs bound sectors of $\,\Xi$ with span no more than $\pi$. 
We will show that  there is a $\delta > 0$, dependent on $\rho$, such that 
for at least one of these critical points, say $A_H$, the $\delta$-neighbourhood of $A_H$ (the set of robots at distance at most $\delta$ from $A_H$) is eventually empty. 


The critical observation is captured in a lemma that states that that if some robot $\,\bZ$ has 
neighbours at distance that is some sufficiently large fraction of $r_H$,  
then its motion cannot take it too close to $A_H$. 
From this it follows immediately that if $\,\bZ$ continues to have at least one such distant neighbour then $\,\bZ$ must become and stay well removed from $A_H$.

\begin{lemma}
\label{lem:separation}
If $\,V_{\bZ} \ge \zeta\cdot r_H$
then any $\xi$-rigid motion of $\,\bZ$ takes it to a point at distance at least\, 
$\left(\frac{\zeta}{80 (1+ 1/ \xi)^{1/2}}\right)^4 r_H$
from $A_H$.
\end{lemma}

\begin{proof}
Refer to Figure~\ref{fig:HullShrinkage}.left.
Let $d = \left(\frac{\zeta}{80 (1+ 1/ \xi)^{1/2}}\right)^4 r_H$, consider a circle $\Gamma_d(A_H)$ of radius $d$ centered at $A_H$.
Let $y_d$ be the length of the chord $PP'$ (shown in blue) tangent to $\Gamma_d(A_H)$ and normal to $A_H O_H$, and let $R_d$ denote the portion of $\,\Xi$, containing point $A_H$, between the arc $PP'$ and the chord $PP'$.  Note that $y_d = 2 \sqrt{r_H^2 - (r_H - d)^2} <  2 \sqrt{2 r_H d}$.

First suppose that the motion of $\,\bZ$ is $1$-rigid. 
If the target destination $T$ satisfies $|A_H T| \le d$
(red sector) then, since $T$ is the midpoint of safe region centers of two (or one) distant neighbours of $\,\bZ$, the safe region center $C_W$ associated with at least one extreme distant neighbour $\bW$ of $\,\bZ$ must lie in $R_d$.

As $C_W$ lies in $R_d$ and $C_W$ lies on $\overline{ZW}$, there can be two cases: either (i) $Z$ lies in $R_d$, or (ii) $W$ lies in $R_d$.
In the first case, we have $|Z C_W| < y_d$, $|ZW|\le 8|Z C_W|$, and thus $V_{\bZ} < 2 |Z W| < 16 y_d < 32 \sqrt{2 r_H d}$.
%
In the second case, if $W$ lies in $R_d$, then $|C_W W| < y_d$, and $|ZW|=V_{\bZ}/8+|C_W W|<V_{\bZ}/8+y_d$.
Thus, as $\bW$ is a distant neighbour of $\bZ$, $V_{\bZ} < 2 |ZW| < 2(V_{\bZ}/8+y_d)$, and hence $V_{\bZ}<8/3 y_d<16/3\sqrt{2r_H d}$.

If the motion of $\,\bZ$ is $\xi$-rigid its actual destination $T'$ satisfies $|Z T'| \ge \xi |Z T|$. Let $|A_H T'| \le d$.
Consider again the two cases: $Z$ lies in $R_d$, and $W$ lies in $R_d$.
%
In the second case, when $Z$ lies outside $R_d$, then $T$ must lie in $R_d$ (as $T'$ lies on $\overline{ZT}$). 
Hence $|A_H T| \le y_d$ and, by our earlier analysis (for $\xi=1$), $V_{\bZ} < 
32 \sqrt{2 r_H (2 \sqrt{2 r_H d})}$. 
On the other hand, if $Z$ lies in $R_d$ then $|Z T'| \le y_d$ and hence $|Z T| \le y_d/ \xi$.  
Thus, 
$|A_H T| \le |A_H Z| + |ZT| < y_d + y_d/ \xi = y_d(1 + 1/ \xi)$
and, by our earlier analysis,
\begin{align*}
V_{\bZ} & < 32 \sqrt{2 r_H y_d (1 + 1/ \xi)}
< 32 \sqrt{4 r_H \sqrt{2 r_H d} (1 + 1/ \xi)}\\
&<80 (1+ 1/ \xi)^{1/2} r_H^{3/4} d^{1/4}=80 (1+ 1/ \xi)^{1/2} (d/r_H)^{1/4}r_H = \zeta r_H \,. \qedhere
\end{align*}
\end{proof}
\noindent It follows directly from Lemma~\ref{lem:separation} that permanent separation from $A_H$ is contagious:

\begin{lemma}
\label{lem:separation2}
If robot $\,\bZ$ has a neighbour $\bX$ whose distance from $A_H$ remains at least $\mu r_H$ then 
$\,\bZ$ must itself become and remain at distance at least 
$\left(\frac{\mu}{240 (1+ 1/ \xi)^{1/2}}\right)^4 r_H$
from $A_H$.
\end{lemma}

\begin{proof}
By Lemma~\ref{lem:separation}, while $V_{\bZ} \ge (\mu/3) r_H$, $\,\bZ$ must remain at distance at least 
$\left(\frac{\mu}{240 (1+ 1/ \xi)^{1/2}}\right)^4 r_H$
from $A_H$.  On the other hand, if $V_{\bZ} < (\mu/3) r_H$ then, since $\bX$ is a neighbour of $\,\bZ$, $\,\bZ$ has distance at least $(2\mu/3) r_H$ from $A_H$. But any movement of $\,\bZ$ (taking it to a position within the convex hull of its neighbours) must leave it no further than $(\mu/3) r_H$ from its current position, and hence at least $(\mu/3) r_H$ from $A_H$. 
\end{proof}

\begin{figure}[tbp]
\centering
\includegraphics[page=1]{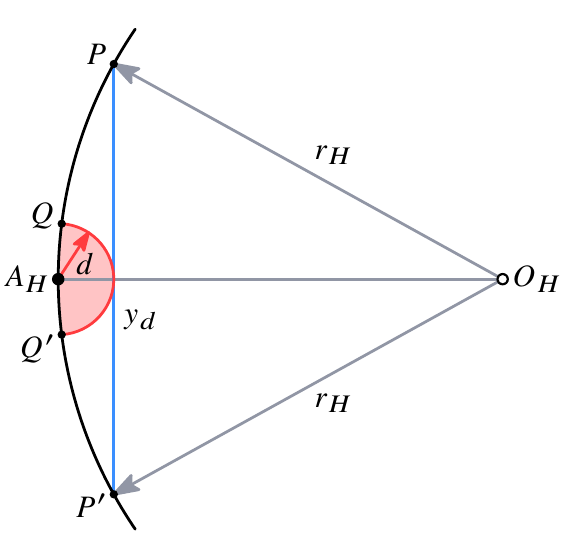}
\hfill
\includegraphics[page=2]{newfig/congregation-lemma.pdf}
\caption{Left: the radius of the red circle is $d$, and the length of the chord $\overline{PP'}$ (in blue) is $y_d$. Right: The length of the chord $\overline{QQ'}$ (in green) is $2z_d$, and the distance from $A_H$ to $\overline{QQ'}$ is $x_d$.
}
\label{fig:HullShrinkage}
\end{figure}

Suppose now that we have chosen a time $t$ after which
the strong neighbour graph no longer changes and the hull radius lies in the range $[\rho, 3\rho/2)$. Consider robots in $\Gamma_\delta(A_H)$, the $\delta$-neighbourhood of vertex $A_H$.
Consider some time $t'>t$ after which all robots that will eventually lie and remain outside of $\Gamma_\delta(A_H)$ have done so. Further, imagine that at time $t'$ one or more robots, including $\,\bZ$, lie inside $\Gamma_\delta(A_H)$. It will suffice to argue that, assuming motion is $\xi$-rigid,  $\,\bZ$ must move and remain outside $\Gamma_\delta(A_H)$. 


There are two cases to consider depending on the composition of 
the set of neighbours of $\,\bZ$.
\begin{enumerate}[(i)]
    \item All other robots are neighbours of $\,\bZ$.
    In this case, either (i)~at some subsequent point in time, either one of $\Gamma_\delta(A_H)$, $\Gamma_\delta(B_H)$ or $\Gamma_\delta(C_H)$ contains no robots (in which case nothing remains to be proved), or (ii)~%
$\,\bZ$ must continue to have a neighbour at distance at least 
$r_H - \delta$, that is $V_{\bZ} \ge r_H - \delta
\ge r_H/2$, provided that $\delta \le r_H/2$. Hence, by Lemma~\ref{lem:separation}, $\,\bZ$ too must move to, and remain at, a position at least 
$\left(\frac{1}{160 (1+ 1/ \xi)^{1/2}}\right)^4 r_H$
from $A_H$.
    
    \item $\,\bZ$ has at least one other robot that is not a strong neighbour. In this case, since the strong neighbour graph is connected, there must be a strong neighbour $\bX$ of $\,\bZ$ that has a strong neighbour $\bY$ that is not a strong neighbour of $\,\bZ$. 
Since $\bY$ and $\,\bZ$ are not strong neighbours their separation must continue to exceed $V/4$, and hence the distance from at least one of them to $\bX$ must continue to be at least $V/8$ which, by the connectivity of the strong neighbour graph, is at least $r_H/(8n)$.
Thus, 
by Lemma~\ref{lem:separation}, $\bX$ itself 
must move to, and remain at, a position at least 
$\left(\frac{1}{640n (1+ 1/ \xi)^{1/2}}\right)^4 r_H$
from $A_H$.
It follows, by  Lemma~\ref{lem:separation2}, that $\,\bZ$ too
must move to, and remain at, a position at least 
$\left(\frac{1}{640n (1+ 1/ \xi)^{1/2}}\right)^{16} \left(\frac{1}{240 (1+ 1/ \xi)^{1/2}}\right)^4 r_H$
from $A_H$.
\end{enumerate}

Hence, if $\delta < \left(\frac{1}{640n (1+ 1/ \xi)^{1/2}}\right)^{16} \left(\frac{1}{240 (1+ 1/ \xi)^{1/2}}\right)^4 r_H$, $\,\bZ$ must eventually, and permanently, vacate $\Gamma_{\delta}(A_H)$.
This contradicts our assumption that at time $t'$ all robots that will permanently vacate $\Gamma_{\delta}(A_H)$ have done so. Thus, we can assume that at time $t'$, and thereafter, $\Gamma_{\delta}(A_H)$ contains no robots.

This means, of course, that $CH_{t'}$, the convex hull of the robot positions at time $t'$, has no vertex closer than $\delta$ to $A_H$. It follows from this that the perimeter of $CH_{t'}$  is less than $\Lambda_t$, the perimeter of $CH_{t}$, by an amount $\lambda$ that is independent of $\varepsilon$.

\begin{lemma}
If the robot positions at time $t'$ are all outside of $\,\Gamma_d(A_H)$ then the perimeter of $\,CH_{t'}$  is at most  
$\Lambda_t - \frac{d^3}{4 r_H^2}$. 

\end{lemma}

\begin{proof}
Since $A_H$ lies on the boundary of $CH_t$, and all robot positions at time $t'$ lie outside of $\Gamma_d(A_H)$, it follows that the perimeter of $CH_{t'}$ is at least $2(d-z_d)$ shorter than the perimeter of $CH_t$, where $2 z_d$ is the length of the chord $QQ'$ (see Figure~\ref{fig:HullShrinkage}.right).
Let $x_d$ be the distance from $A_H$ to $\overline{QQ'}$.
It is straightforward to confirm that
$d-z_d = x_d^2 / (d + z_d) > x_d^2 / (2d)$ and (by similar triangles) $x_d = d^2 / (2 r_H)$, and hence
$2(d-z_d) > d^3 /(4 r_H^2)$.
\end{proof}

As previously noted, this leads to a contradiction (of our assumption of non-convergence) if we choose $\varepsilon < \frac{d^3}{4 r_H^2}$ and $t > t_{\varepsilon}$.
Therefore the sequence of nested convex hulls converge to a point, and our algorithm solves {\sc Point Convergence} problem.



\section{Extensions/Generalizations}\label{sec:Extensions}


\subsection{Error tolerance}\label{subsec:error}
\begin{itemize}
    \item It is clear that any absolute error in distance measurements makes convergence to a point impossible, even in a fully synchronous scheduling environment. However our algorithm 
    can be modified to tolerate a modest amount of relative error in distance measurements
   The only concern with error in distance measurements 
   is that it might lead $V_{\bZ}$, the perceived distance to the closest neighbour of $\,\bZ$, to become an overestimate of $V$. If the error is bounded by some small fraction $\delta$ of the true distance, then this can be avoided by simply scaling the perceived distance by a factor of $1/(1 + \delta)$.
    
    \item In the absence of any global information on the initial configuration, even an arbitrarily small amount of absolute error in angle measurements makes convergence infeasible for any algorithm, even in fully synchronous scheduling environment. This follows by observing that the presence of such error would make it impossible to distinguish some non co-linear triple of points, separated by distance $V$, with a co-linear such triple. Given this, any algorithm would be forced to refrain from moving the midpoint of such a co-linear triple (at the risk of a catastrophic separation). This in turn would lead to completely frozen activity on a suitably collection of robots configured as a regular polygon with vertex separation $V$.
    
  
    \tab However, our algorithm can be modified to tolerate a modest amount of  a kind of relative error in angle measurements (that preserves sidedness with respect to lines through neighbouring points).

  \tab To make this concrete, consider the error in angle measurements that would arise from a small perturbation of a robot's local coordinate system: a \emph{symmetric distortion} of some reference coordinate system is a continuous bijection $\mu: [0, 2\pi]) \rightarrow [0, 2\pi)$ with $\mu(\theta + \pi) = \mu(\theta) + \pi$, for all $\theta \in [0, \pi)$.  
    Such a distortion has \emph{skew} bounded by $\lambda < 1$ if 
    $(1- \lambda)\xi
    \le \mu(\theta + \xi) - \mu(\theta)
    \le (1+ \lambda)\xi$, for all $\xi \in (0, \pi)$.
    
    \tab Observe that a symmetric distortion with skew $0$ is just a rotation or reflection of the local coordinate system.
    That a permissible distortion is a continuous bijection seems natural; that it should also be symmetric is also realistic in our context, since otherwise it is possible to construct (as above) a configuration with robots located at the vertices of a regular polygon with $V$ length edges, that fails to converge.
    
    \tab Our algorithm can be modified to tolerate a symmetric distortion of its local coordinate system with small skew. The intuition here is to reduce the safe region of robot $\bY$ with respect to a distant neighbour $\bX$ by making it instead the intersection of the safe regions that would be associated with all possible true locations of $\bX$. 
    
    \item 
   It should be clear that absolute error, measured as the deviation of the actual destination from the trajectory to the intended destination, cannot be tolerated by any convergence algorithm, even in the fully synchronous model. 
    In fact, even linear relative error (error that grows linearly with the distance travelled) cannot be tolerated (see Figure~\ref{newfig:MotionError3}).
    
    \begin{figure}[tbp]
\centering
\includegraphics{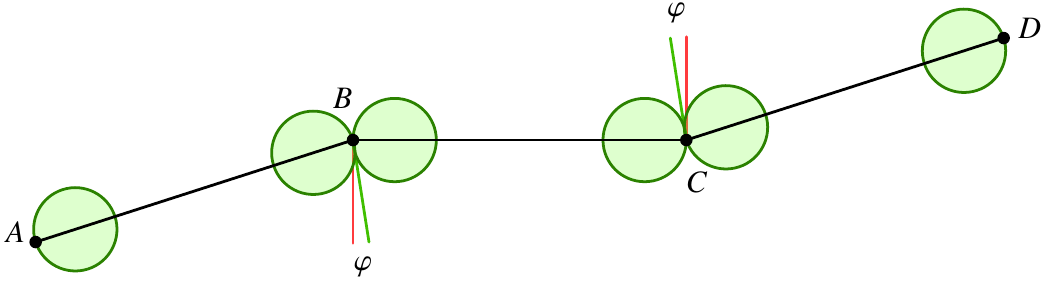}
\caption{If the relative motion error is greater than $\tan \varphi$ then $B$ can move to the left of the normal to $BC$ (red) and $C$ can move to the right, resulting in separation if $|BC|=V$.}
\label{newfig:MotionError3}
\end{figure}
    
   \tab However, our algorithm can be modified to tolerate relative motion error that grows quadratically (specifically error is $O(d^2/V)$, where $d$ is the distance travelled). To see this, recall that our algorithm chooses as its destination point the midpoint of the region formed by the intersection of the safe regions of a robot's extreme distant neighbours. If the angle formed by these neighbours is bounded away from $\pi$ then there is fixed faction of the intended trajectory that, even when followed with quadratic error, will remain within the intersection of the safe regions. 
   
   \item Our algorithm can clearly tolerate premature stop failures, provided these do not accumulate in a way that violates the activation fairness condition. In fact,  a single crash (or ``fail-stop'') fault, where a robot crashes and remains stationary thereafter, can also be tolerated. In this case the remaining robots converge to the location of the crashed robot.

\end{itemize}

\subsection{Visibility}
Our algorithm works correctly even if the visibility region is open (i.e. to be visible a neighbouring robot must be at distance strictly less than $V$). To see this, it suffices to observe that in this case, $V_{\bZ}$, the perceived distance to the furthest neighbour, is always strictly smaller than $V$. 

Furthermore, can be modified without difficulty to work correctly even if the visibility radius $V$ may differ for different robots, provided that (i) the initial mutual visibility graph is connected, and (ii) individual (unknown) visibility radii could differ by at most some small (known) constant factor.

In the event that the visibility radius $V$ exceeds the diameter of the initial configuration, it follows from the hull-diminishing property that all robots will remain mutually visible throughout the computation. Since our congregation argument continues to hold even under an \async scheduler, we conclude that our $1$-\async algorithm will guarantee convergence in this case, {\em without any need for multiplicity detection}.


\subsection{Further relaxations}

\subsubsection{General initial configuration.}
It is worth noting that if the initial configuration is not connected our algorithm guarantees that every connected subset of robots will converge to a single point. Furthermore, noting that robots in any component embedded in an interior face of the visibility graph associated with the initial configuration must ultimately converge to the same point as the robots defining that face, it is clear that connectivity of the initial configuration is not a necessary condition for our algorithm to solve {\sc Point Convergence}.

\subsubsection{Three (and higher) dimensions.}
Our algorithm admits a very natural generalization for robots configured in three (and higher) dimensions (as in~\cite{BCF20,YUKY17}). Specifically, safe regions are generalized to higher dimensional disks, with the same center and radius. While it seems clear that both our visibility preservation and incremental congregation arguments can be extended, the geometric details are more intricate. We leave the full resolution of these details to future work.
%
%
%

\section{Impossibility of visibility preservation in \async model, assuming modest error in perception, but no error in motion}\label{sec:ImpossibilityA}

In this section, we will show how to associate with any modestly error-tolerant control algorithm an initial connected configuration of robots, together with an adversarial scheduler in the \async model 
(in fact, the \nesta model) that produces a robot configuration with two or more linearly separable connected components in its associated visibility graph.
Specifically, we consider algorithms, like the one developed in preceding sections, and like those that we might reasonably expect to be used in practice, for robots whose perceptual apparatus is accurate up to 
(i) small relative error in 
distance measurements (the perceived distance $\tilde{d}$ to any neighbour of a robot is accurate only to within a constant $\delta>0$ times the true distance $d$), and
(ii) small relative errors in angle measurements
of the kind that arise from symmetric distortions of local coordinate systems, with bounded skew 
(the perceived angle $\tilde{\theta}$
between any two neighbours is accurate only to within some constant $0< \lambda <1$ times $\pi - \theta$, where $\theta$ is the true angle). 
The control algorithm is free to assume a fixed visibility threshold (assumed to be 1), and that motion is rigid and free of error. 

The first assumption allows us to specify 
an initial configuration of robots so that the true distance of neighbouring robots, in both this initial configuration and all subsequent configurations, is sufficiently close to the visibility threshold that neighbours could be perceived as always being exactly at that threshold. 
The second assumption allows us to conclude that, while the perceived sidedness of the observing robot with respect to the line joining any pair of its neighbours agrees with reality, there is sufficient uncertainty to guarantee that robots that are not yet co-linear with their neighbours must move.
(Otherwise, as we will show, it would be possible to construct a local configuration that likewise has no associated forced move, but which could be replicated into a global configuration that, for this reason, fails to converge).

It follows that the {\sc Cohesive Convergence} problem is not solvable in the \async 
or \nesta 
model. This in turn implies that the  {\sc Point Convergence} problem is not solvable in the \async model, using any hull-diminishing algorithm.
In fact, even if an algorithm satisfies the weaker condition of being \emph{coherent}, in the sense that it never prescribes a motion for a robot that will immediately disconnect the visibility graph induced on its neighbours, then it must fail to converge from the disconnected configuration forced by the adversary that we describe.

We conclude that the {\sc Point Convergence} problem illustrates a separation in the power of bounded vs. unbounded asynchrony,
for 
natural (hull-diminishing or coherent), 
modestly error-tolerant algorithms.
Note that, it seems unavoidable to insist on some measure of error tolerance in framing this result. In fact, a proof that something is impossible for error-free algorithms with rigid motion seems well out of reach.

\subsection{The initial configuration}

We describe the initial configuration in terms of a parameter 
$\psi >0$
that we will subsequently set to be sufficiently small. We take the visibility threshold $V$ to be $1$.

The initial robot configuration consists of three robots $\bX_A$, $\bX_C$ and $\bX_B$, located at positions $A= (0,0)$, $C=(-\frac{1}{\sqrt{2}},-\frac{1}{\sqrt{2}})$ 
and $B=(1,0)$ respectively, and $n-3$ additional robots, 
$\bX_1, \bX_2, \ldots \bX_{n-3}$, positioned at points $P_1, P_2, \ldots, P_{n-3}$ spaced at distance $1$, along a discrete spiral tail starting with the edge $AB$. 
The turn angle between the chord $\overline{A P_{i-1}}$ and the segment $\overline{P_{i-1} P_i}$, taking $P_0 = B$, is fixed at $\psi$, for all $i \ge 1$. 
(Figure~\ref{fig:Example}.left illustrates such a  configuration%
). 

The total number of robots, $n$, is chosen to be sufficiently large that the angle between the chords $\overline{A P_0}$ and $\overline{A P_{n-3}}$ is about $3 \pi /8$ (so the line through $A$ and $P_n$ bisects the angle $\angle(CAB)$). Of course, we need to argue that such an $n$ exists: for this we observe that even though the angle $\gamma_i$ between successive chords $\overline{A P_{i-1}}$ and $\overline{A P_i}$ decreases with $i$, the sum $\sum_i \gamma_i$ diverges.
To see this, let $d_i$ denote the length of the chord $\overline{A P_i}$. We know, from the cosine law, that 
$d_i^2 = d_{i-1}^2 +1 -2d_{i-1} \cos(\pi - \psi)$. 
Hence, when $\psi >0$,
\[
\begin{aligned}
    d_i - d_{i-1}  
        & = \frac{1 +2 d_{i-1} -2 d_{i-1}(1- \cos \psi)}
                {d_i + d_{i-1}}
        > \frac{1 +2 d_{i-1} -2 d_{i-1}(1- \cos \psi)}
                {1 +2 d_{i-1}}\\[1em]
        & = 1 - \frac{2 d_{i-1}(1- \cos \psi)}
                {1 +2 d_{i-1}}
        \ge 1 - \frac{2 d_{i-1}(\psi^2/2)}
                {1 +2 d_{i-1}}
        > 1 - \psi^2/2.
\end{aligned}
\]

Since $d_0 =1$, and $d_i - d_{i-1} \le 1$ (by the triangle inequality), it follows by induction on $i$ that
$i(1 - \psi^2/2) < d_i < i$, for $1 \le i \le n-3$.

It is easy to confirm that the perpendicular distance from $P_i$ to the line through $A$ and $P_{i-1}$ equals both $\sin \psi$ and $d_i \sin \gamma_i$. 
Hence $\gamma_i \ge \sin \gamma_i = \frac{\sin \psi}{d_i} \ge \frac{\sin \psi}{i}$. 
It follows that $\sum_{i=1}^{n-3} \gamma_i \ge \sin \psi \ln (n-3)$. 
Thus it will suffice to choose $n$ large enough that 
$\sin \psi \ln (n-3) \ge 3 \pi /8$, i.e.,  
$n \ge 3+ e^{\frac{3 \pi}{8 \sin \psi}}$.

\begin{figure}[t]
\centering
\includegraphics[page=1]{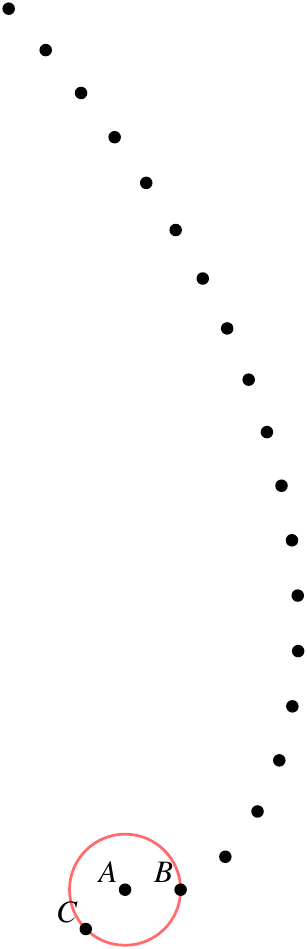}
\hfill
\includegraphics[page=2]{newfig/ex.pdf}
\hfill
\includegraphics[page=3]{newfig/ex.pdf}
\caption{Left: initial robot configuration; Center: approximate trajectories to final configuration; Right: typical intermediate configuration}
\label{fig:Example}
\end{figure}

\subsection{Adversarial scheduler strategy}

The strategy of the adversarial scheduler is to first activate the robot $\bX_A$ 
forcing it (as we will argue) to plan a move to a point $A'$ some non-zero distance $\zeta$ away from $A$ into the sector $CAB$ (by the symmetry of the local configuration, we can assume that $|A'C| \le |A'B|$; for ease of description, we will assume that these distances are equal).

Next, and before the robot $\bX_A$
actually begins its move, the scheduler begins a long sequence of activations of the robots $\bX_0, \ldots, \bX_{n-3}$ (where $\bX_0 = \bX_B$), with the goal of moving, for each $i$, $1 \le i \le n-3$, in succession, all of the robots $\bX_0, \bX_1, \ldots, \bX_{i-1}$ onto the chord $AP_i$, without changing their initial distance from $A$ by more than an amount that is $\Theta(\psi^2)$. If successful in this task then, by choosing $\psi$ to be sufficiently small,
all robots will follow trajectories that approximate circular arcs, centered at $A$, from their initial position to their final position on the chord $\overline{AP_n}$ (see Figure~\ref{fig:Example}.center and Figure~\ref{fig:Example}.right). 
Furthermore, if $\psi$ is chosen to be sufficiently small
relative to $\zeta$, the scheduler will have succeeded in breaking the visibility between the robots $\bX_A$ and $\bX_B$. 

The argument that such an activation sequence exists involves two things: (i) showing that specified robots must move, when activated, and (ii) showing that the motion must keep their distance from $A$ the same, up to an additive factor that accumulates to something that is $O(\psi^2)$. 

The space between successive chords defines what we call a sliver (see Figure~\ref{fig:sliver}). 
The transitions of robots from one chord to the next, flattening the associated sliver, involves the (essential) collapse of a succession of thin triangles, each of which involves the relocation of a robot to become essentially co-linear with its neighbours at almost distance one. 
By ``essentially co-linear'' we mean that the angle formed by the neighbours is in $(\pi - \psi/2^n, \pi]$. This slight slack permits termination of the collapse process and ensures that when the entire relocation is complete the sum of all turn angles on the chain is at most $n \psi/ 2^n$ and hence all of the robots lie arbitrarily close to the chord $A P_{n-3}$.

\begin{figure}[htbp]
\centering
\includegraphics{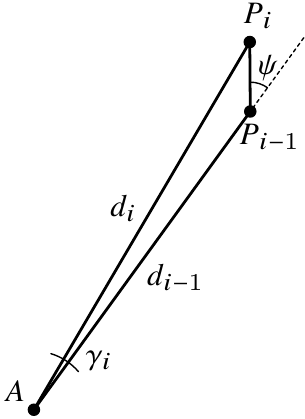}
\caption{Initial sliver.}
\label{fig:sliver}
\end{figure}

\subsubsection{Forced motion}

Imagine that a triple of robots is in a local configuration like that illustrated in Figure~\ref{newfig:TriangleCollapse}. Under what conditions can we assume that if activated the robot at position $Q$ must move? Our assumption about error in distance perception is such that if the distances between $Q$ and its two neighbours are both in the range $(1-\delta, 1]$, then they could both be perceived as being $1$.
Clearly, this would preclude the algorithm from choosing not to move the robot at $Q$ (when activated) if in addition the angle $\phi$ is of the form $2\pi / k$, for some integer $k$; otherwise the algorithm would fail to converge if started from a configuration with robots located at the vertices of a regular $k$-gon with unit length sides. Thus, if the error in angle perception is such that the angle $\phi$ could be perceived to have this special form, some motion of the robot at $Q$ is forced.

Our assumption about error in angle perception is such that the true turn angle $\phi$ could be perceived as anything in the range $[\phi(1-\lambda), \phi]$, for some positive constant $\lambda <1$. Since this range might not contain a special angle when $\phi$ is very small, we need a more versatile argument to see that motion can be forced. 
Fortunately, 
it suffices to observe that, 
if integer $M$ is greater than $\frac{4\pi}{\lambda \phi}$, then there must exist angles $\frac{2\pi i}{M}$ and $\frac{2\pi (i+1)}{M}$ both in the range $[\phi(1-\lambda), \phi]$. 
It must be the case that if the algorithm chooses not to move the robot  at point $Q$ (when activated) when the turn angle is perceived to be $\frac{2\pi i}{M}$ then it must move the robot at point $Q$ (when activated) when the turn angle is perceived to be $\frac{2\pi (i+1)}{M}$. 
Otherwise, we could construct another polygonal placement with $2M$ edges that alternate a 
$\frac{2\pi (i+1)}{M}$ counter-clockwise turn with a 
$\frac{2\pi i}{M}$ clockwise turn, from which the algorithm would fail to move, and hence fail to converge.

\subsubsection{Collapsing thin triangles}

\begin{figure}[htbp]
\centering
\includegraphics{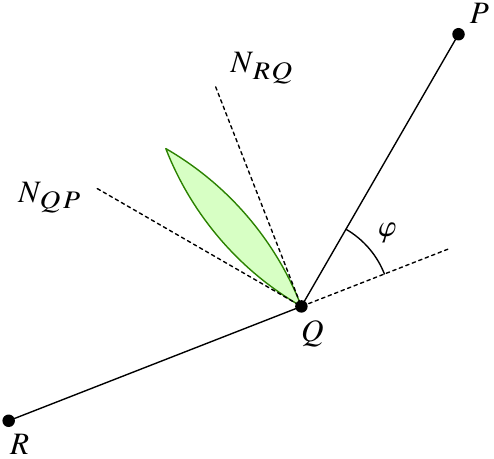}
\caption{}
\label{newfig:TriangleCollapse}
\end{figure}

It is not enough to know that motion can be forced by an adversary. We also need to demonstrate that (i) motion is {\em productive}, in the sense that it eventually leads to essential co-linearity
\footnote{To simplify the argument, we treat ignore the fact that ``essential co-linearity'' is not exactly the same  as co-linearity. This is justified by the fact that, as we will demonstrate, robots can be moved to become arbitrarily close to co-linear, and the extra work involved in carrying insignificant error terms would not make the argument stronger or more illuminating in any measurable way.}
with neighbours, and (ii) motion preserves distance to point $A$, up to some small additive error, proportional to angle $\phi$ and the distance moved.

Here again, it helps to consider a triple of robots is in a configuration like that illustrated in Figure~\ref{newfig:TriangleCollapse}, where the distances $|RQ|$ and $|QP|$ are both $1$, as perceived by the robot at $Q$. 
We know that the robot at location $Q$ (when activated) is forced to move, and its intended destination must lie within the green lens; otherwise, such a move would separate this robot from one of its neighbours, following which convergence would be impossible if this were the full initial configuration.


Any motion only leads to a configuration with smaller turn angle $\phi$. (This is true even if the motion takes $Q$ beyond the midpoint of the lens). 
Since the new turn angle $\phi'$ could continue to be perceived as being in the range
$[\phi(1-\lambda), \phi]$ as long as 
$\phi' > \phi/(1 + \lambda)$,
it follows that repeated activation of the same robot will keep it inside the lens and eventually bring it to a position $Q'$ of essential co-linearity with its neighbours, provided that the distances to the neighbouring robots continues to fall in the range $(1-\delta, 1]$. 

Any point $Q'$ of  co-linearity within the lens has distance $\ell$ at most $2 \sin (\phi/4)$, which is at most $\phi/2$, from 
$Q$. Furthermore, the distance from $Q'$ to both of the normals $N_{RQ}$ and $N_{QP}$ is at most $\ell \sin \phi$.
From this it follows that the distance from $Q'$ to any point at or beyond $R$ on the ray from $Q$ through $R$ (or any  point at or beyond $P$ on the ray from $Q$ through $P$) does not change by more than $\ell \sin \phi$,
which is at most $\phi^2/2$.

Thus, it suffices to assume that the distances to the neighbouring robots starts in the range
$(1-\delta + \phi^2/2, 1]$. Furthermore, once the robot has reached a position of essential co-linearity with its neighbours, its distance from point $A$, the base point of our global configuration, has not changed by more than $\phi^2/2$. 

\subsubsection{Flattening slivers}
A sliver (see Figure~\ref{fig:Collapse}.top\_left, for a generic example) is iteratively flattened by
first (i) moving the endpoint $Q$ of the first chord until it is (essentially) co-linear with its neighbours $R$ and $P$, collapsing the associated triangle (see Figure~\ref{fig:Collapse}.top\_right), as described in the preceding subsubsection,  then (ii) recursively flattening the newly created sliver $ARQ'$ (see Figure~\ref{fig:Collapse}.bottom\_left), and
finally (iii) continuing to flatten the semi-flattened sliver $AQ'P$ (see Figure~\ref{fig:Collapse}.bottom\_right). 

\begin{figure}[htbp]
\centering
\includegraphics[page=1]{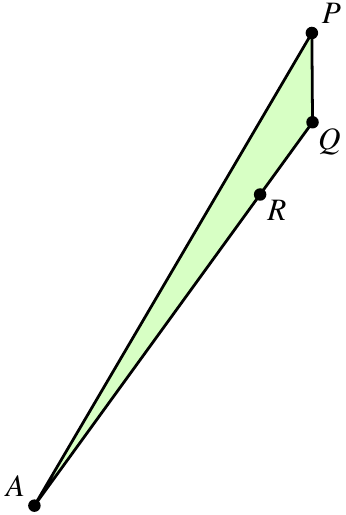}
\hfill
\includegraphics[page=3]{newfig/collapse.pdf}
\hfill
\includegraphics[page=4]{newfig/collapse.pdf}
\hfill
\includegraphics[page=2]{newfig/collapse.pdf}
\caption{Recursive collapse of a sliver.}
\label{fig:Collapse}
\end{figure}

It remains to argue that the change in the distance between the location of robot $\bX_j$ and $A$, even though it may accumulate through a long sequence of triangle collapses, does not sum to more than $4 \psi^2$. To see this we observe that:
\begin{enumerate}[(i)]
    \item all slivers formed by successive chords in the initial placement have turn angle $\psi$;
    
    \item the successively narrower slivers that arise in flattening any initial sliver have progressively smaller turn angles;
    
    \item the sliver that arises after the first triangle collapse (the blue sliver in Figure~\ref{fig:Collapse}.bottom\_left) has a turn angle that is exactly one half of that its parent sliver; and more generally the slivers that arise in the $t$-th level of recursive sliver flattening have a turn angle that is   a fraction $1/2^t$ of that of its parent sliver.
\end{enumerate}
It follows that the change in distance from $A$ associated with the motion of robot $\bX_j$ in the course of moving it from its location on chord $A P_{i-1}$ to its location on chord $A P_i$ is at most $\gamma_i d_j \frac{\psi}{2^{i-j-1}}$.  
Hence the total deviation from its initial distance $d_j$ is at most 
$\sum_{i>j} \gamma_i d_j \frac{\psi}{2^{i-j-1}} 
< \gamma_j d_j \psi \sum_{i>j} \frac{i}{2^{i-j-1}} 
< 2 \gamma_j d_j \psi 
< 4 \sin \gamma_j d_j \psi
= 4 \frac{\sin \psi}{j} d_j \psi
< 4 \psi^2$.

\section{Summary and Conclusion}\label{sec:Conclusion}

We have studied the well-known {\sc Point Convergence} problem for autonomous mobile robots with bounded visibility. 
Connectivity, realized through visibility, plays a fundamental role in all known algorithms for {\sc Point Convergence} in this setting:
it is typically assumed that the visibility graph of the initial robot configuration is connected, and 
a central algorithm design constraint is the preservation of visibility between pairs of robots.
Indeed all previous algorithms solve the apparently more restricted problem, which we call {\sc Cohesive Convergence}, that requires the visibility between all initially mutually visible robot pairs to be preserved indefinitely.

We propose a new algorithm for 
{\sc Cohesive Convergence} 
designed for robots operating within the $k$-\async scheduling model, for any constant $k$, a significantly more inclusive scheduling environment than that assumed by its predecessors. In addition, our algorithm makes comparatively weak assumptions, including non-rigid movements and no knowledge of the visibility radius $V$, about robot capabilities. 
Another significant distinguishing feature is our algorithm’s tolerance for 
limited imprecision in the distances and angles 
perceived by the robots, and limited relative error in the realization of intended motions.
Although formulated for robots moving in two dimensions, our algorithm has a very natural extension to three (and even higher) dimensions, with somewhat more involved correctness proofs.
As it happens, our algorithm, without modification, also solves {\sc Point Convergence} in the fully asynchronous (\async) scheduling model, when visibility is unbounded, without depending on the ability to detect multiplicities in robot locations.

We also establish a complementary result that {\sc Cohesive Convergence} cannot be solved in the fully asynchronous (\async) scheduling environment, assuming that robots are subject to very limited imprecision in measurements, even assuming exact rigid motion and knowledge of $V$. 
Hence,  {\sc Cohesive Convergence}  demonstrates a separation between the power of autonomous robots operating in a scheduling models with arbitrarily high but bounded asynchrony ($k$-\async), and those operating in the presence of unbounded asynchrony (\async).
As a special case, this also provides a positive answer to a long-standing open question: whether robots operating in a semi-synchronous (\ssync) scheduling environment  are strictly more powerful than those operating in a fully asynchronous (\async) scheduling environment.

Our results leave open several natural questions. Among these we note:
\begin{enumerate}
\item
Are there problems that serve to distinguish the power of robots operating in the $k$-\async and  \ssync scheduling environments?
\item
Is ({\sc Cohesive}) {\sc Point Convergence} solvable in the fully asynchronous (\async) scheduling environment if there is no error in perception or motion, or if error is restricted to some arbitrarily small relative quantity arising from motion alone?
\item
From our investigation and from the literature, it turns out to be natural to consider {\sc Cohesive Convergence} rather than simply {\sc Point Convergence}. 
However, it is not clear whether hull-diminishing or coherent algorithms, as natural as they may be, can actually be avoided. 
In general, it remains open to understand whether the resolution of {\sc Point Convergence} necessarily solves  {\sc Cohesive Convergence} as well.
\end{enumerate}

As future work, it would be interesting to extend our results to consider robots with non-zero extent, as in the \emph{fat} robots studied in~\cite{AGM13,CM15,DCGP09}. Other natural constraints worth investigating include occlusion, as in~\cite{BhGM18,BoKAS21}, and collision avoidance, as in
\cite{PPV15}.\ \\

\noindent {\bf Acknowledgments.} 
We would like to acknowledge the stimulating research  environment 
provided by the Bertinoro Workshop on Distributed Geometric Algorithms, held in August 2019, where
the  seeds of this paper were planted. In particular, we  gratefully acknowledge 
the  discussions with Pierre Fraigniaud at the beginning of this work.
This research project has been supported in part by the Natural Sciences and Engineering Research
Council (Canada) under  Discovery Grants A3583 and A2415, and by the Italian National Group for Scientific Computation GNCS-INdAM.

\bibliographystyle{plainurl}
\bibliography{refs.bib}
\end{document}